\documentclass[12pt, draftclsnofoot, onecolumn]{IEEEtran}

\usepackage{algorithm,amsbsy,amsmath,amssymb,epsfig,bbm,mathrsfs,multirow,amsthm,xcolor,subcaption,bm, cite}
\usepackage[spaces,hyphens]{url}

\usepackage{hyperref}
\hypersetup{
	colorlinks=true,
	linkcolor=black,
	filecolor=black,   
	urlcolor=black, citecolor=black,}

\usepackage{algorithmic}
\newtheorem{theorem}{Theorem}

\DeclareMathAlphabet\mathbfcal{OMS}{cmsy}{b}{n}
\usepackage{array}

\makeatletter
\def\endthebibliography{%
	\def\@noitemerr{\@latex@warning{Empty `thebibliography' environment}}%
	\endlist
}
\makeatother

\begin{document}
	
	\title{Countering Eavesdroppers with Meta-learning-based Cooperative Ambient Backscatter Communications}
	
	\author{Nam H. Chu, Nguyen Van Huynh, Diep N. Nguyen, Dinh Thai Hoang,\\ 
		Shimin Gong, Tao Shu, Eryk Dutkiewicz, and Khoa T. Phan
		\thanks{Nam H. Chu, Diep~N.~Nguyen, Dinh~Thai~Hoang, and Eryk~Dutkiewicz are with the School of Electrical and Data Engineering, University of Technology Sydney, Australia (e-mails: namhoai.chu@student.uts.edu.au, \{diep.nguyen, hoang.dinh, eryk.dutkiewicz\}@uts.edu.au).}
		\thanks{Nguyen Van Huynh is with School of Computing, Engineering and the Built Environment, Edinburgh Napier University, Edinburgh, United Kingdom (e-mail: h.nguyen2@napier.ac.uk).}%
		\thanks{Shimin Gong is with the School of Intelligent Systems Engineering, Sun Yat-sen University, Guangzhou, China. (e-mail: gongshm5@mail.sysu.edu.cn).}
		\thanks{Tao Shu is with the Department of Computer Science and Software Engineering, Auburn University, Auburn, AL 36849. (e-mail: tshu@auburn.edu).}
		\thanks{Khoa~T.~Phan is with School of Engineering and Mathematical Sciences, Department of Computer Science and Information Technology, La Trobe University, Melbourne, Australia (e-mail: K.Phan@latrobe.edu.au).}%
		\thanks{Preliminary results in this paper were presented at the IEEE WCNC Conference, 2023~\cite{huynh2023defeating}.}
	}
	\maketitle
	
	\begin{abstract}
		This article introduces a novel lightweight framework using ambient backscattering communications to counter eavesdroppers.  
		In particular, our framework divides an original message into two parts.
		The first part, i.e., the active-transmit message, is transmitted by the transmitter using conventional RF signals. 
		Simultaneously, the second part, i.e., the backscatter message, is transmitted by an ambient backscatter tag that backscatters upon the active signals emitted by the transmitter. 
		Notably, the backscatter tag does not generate its own signal, making it difficult for an eavesdropper to detect the backscattered signals unless they have prior knowledge of the system. 
		Here, we assume that without decoding/knowing the backscatter message, the eavesdropper is unable to decode the original message.
		Even in scenarios where the eavesdropper can capture both messages, reconstructing the original message is a complex task without understanding the intricacies of the message-splitting mechanism.
		A challenge in our proposed framework is to effectively decode the backscattered signals at the receiver, often accomplished using the maximum likelihood (MLK) approach. 
		However, such a method may require a complex mathematical model together with perfect channel state information (CSI).
		To address this issue, we develop a novel deep meta-learning-based signal detector that can not only effectively decode the weak backscattered signals without requiring perfect CSI but also quickly adapt to a new wireless environment with very little knowledge. 
		Simulation results show that our proposed learning approach, without requiring perfect CSI and complex mathematical model, can achieve a bit error ratio close to that of the MLK-based approach.
		They also clearly show the efficiency of the proposed approach in dealing with eavesdropping attacks and the lack of {training data for deep learning models} in practical scenarios.		
	\end{abstract}
	
	
	\section{Introduction}
	\label{Sec:intro}
	
	{Given the broadcast nature of the wireless medium, there exists a potential for illegal tapping of any wireless channel, thereby posing a significant risk to the confidentiality of transmitted information, such as personal details or business secrets. 
		Consequently, ensuring privacy and security in wireless communication systems has consistently emerged as a paramount challenge.}
	Among security threats in wireless communications, the eavesdropping attack is one of the most common types of wireless attacks.
	To {perform} the attack, the eavesdropper usually stays close to the victim system to ``wiretap'' the legitimate wireless channel and acquire {exchanged} information.
	Since the eavesdropper operates passively without introducing noise or {altering} transmit signals, detecting and preventing {eavesdropping} attacks are usually challenging.	  
	
	To deal with eavesdropping attacks, conventional approaches primarily {rely} on implementing encryption at the application and transportation layers~\cite{hu2018covert}. 
	Nevertheless, these approaches possess several {issues} that significantly limit their {practical applications}, especially in resource-constrained devices. 
	{Firstly}, these approaches require additional computing resources to facilitate the encryption and decryption, making them less practical or even infeasible for {resource-limited} devices such as IoT devices\cite{zhang2016secure}. 
	{Secondly}, distributing and managing cryptographic keys also require additional communication resources (e.g., frequencies and {transmission} power) and are very challenging tasks, especially {in} decentralized systems comprising a large number of mobile devices~\cite{lu2020intelligent}.
	Finally, an eavesdropper with sufficient computational capacity can decrypt the encrypted data, especially with the recent advances in quantum computing~\cite{google2015quantum}.  
	In addition, by using side-channel analysis, a strong eavesdropper can defeat many cryptographic schemes, even {those} with very robust schemes~\cite{hu2018covert},~\cite{bash2015hiding}.
	
	Unlike cryptography-based approaches, physical layer security leverages physical characteristics of wireless channels (e.g., the signal strength) to protect information from eavesdroppers without requiring additional distributing and managing of cryptographic keys~\cite{mukherjee2014principles}. 
	In physical layer security, friendly jamming is a popular approach in which artificial {noise is} intentionally {injected} into wireless channels to disrupt eavesdroppers' signal reception~\cite{soltani2018covert, siyari2017friendly, mobini2019wireless}.  
	However, this approach may not always ensure a positive secrecy rate, which is the difference of capacity between the legitimate channel {(from the transmitter to the intended receiver)} and the ``tapped'' channel {(from the transmitter to the eavesdropper)}.
	Moreover, the generated interferences may also severely degrade signal receptions at neighboring legitimate devices, especially in densely populated wireless networks.    
	Recently, friendly jamming and beamforming techniques have been leveraged in cooperative transmissions to counter eavesdropping attacks~\cite{yang2017optimal}.
	In particular, the cooperative transmission uses relays to forward data from a transmitter to a receiver.
	In the cooperative jamming mode, relays can also inject noises (i.e., jamming signals) to confuse eavesdroppers in the relays' coverages.
	Whereas in the cooperative beamforming mode, these relays perform distributed beamforming toward the legitimated receiver to minimize the signal strength at eavesdroppers.  
	The main drawback of friendly jamming- and beamforming-based solutions is that they generally require prior information {of the legitimate channel state information (CSI)} in order to achieve effective protection performance.
	However, acquiring such information is often impractical and challenging in real-world scenarios.
	Moreover, these techniques require additional resources (e.g., energy and computing) for generating jamming signals and performing beamforming.
	
	Given the above, this article leverages the cooperative Ambient Backscatter (AmB) {communications} in which the transmitter is equipped with an AmB tag that can backscatter ambient radio signals to convey information to the receiver~\cite{hoang2020ambient}.
	In the literature, the AmB {communications} {have} been well investigated in various aspects, such as hardware design~\cite{liu2013ambient}, performance improvement~\cite{bharadia2015backfi, parks2014turbocharging}, power reduction~\cite{kellogg2016passive, wang2017fm}, ambient backscatter-based applications~\cite{nikitin2012passive, hoang2016tradeoff}, and security~\cite{li2020secrecy, li2021hardware, li2021physical}.  
	However, utilizing AmB technology to counter eavesdropping attacks has not yet {been} well investigated. 
	For example, the works in~\cite{li2020secrecy, li2021hardware} investigate the security and reliability of AmB-based network with imperfect hardware elements, while the authors in~\cite{li2021physical} analyze the application of physical layer security for AmB communications.
	{Unlike these} works in the literature, our proposed solution exploits the advantages of AmB communications together with {a simple} encoding technique to secure the transmitted information against eavesdroppers.
	In particular, the original message is split into two {parts}: (i) the active message transmitted by the transmitter using conventional active transmissions and (ii) the AmB message transmitted by the backscatter tag using AmB transmissions.
	Then, the receiver reconstructs the original message based on both active and AmB messages.
	Note that the AmB tag operates in a passive manner, meaning it does not actively transmit signals.
	Instead, it utilizes the active signals from the transmitter to backscatter the AmB message without requiring additional power.
	In this way, the AmB message can be transmitted on the same frequency and at the same time as the active message.
	However, the AmB signal strength at the receiver is significantly lower than that of the active signal.
	Hence, the AmB signal can be considered as pseudo background noise for the active signal~\cite{van2018reinforcement, zhang2016secure}. 
	As such, without knowledge about the system in advance (i.e., the settings of AmB transmission), the eavesdropper is even not aware of the existence of the AmB message. 
	We assume that without the AmB message, the eavesdropper can not recover the original message.
	
	{It is worth noting that if eavesdropper is aware of the AmB message, it is still a challenging task to capture and reconstruct the original message.
		Specifically, the values of resistors and capacitors in the AmB circuit are different for different backscatter rates.
		As such, even if the eavesdropper deploys the AmB circuit but do not know the exact backscatter rate, it still cannot decode the backscatter signals~\cite{liu2013ambient}.
		In addition,} in the worse case in which the eavesdropper knows the exact backscatter rate in advance, it 
	still does not know how to construct the original message based on the active and AmB messages since the message encoding technique is unknown to them.
	To quantify the security of our proposed anti-eavesdropping solution in the worst case, this paper considers the guessing entropy {metric}~\cite{boztas1997comments}.
	It is {worth noting} that due to the lower rate of AmB transmission compared to that of the active transmission, the AmB message's size is smaller than that of the active message. 
	Thus, in a system with sufficient computing and energy capacities for encryption/decryption, the AmB message can be used to carry the encryption key, while the active message carries the encrypted message.
	
	To detect backscattered signals at the receiver, traditional methods often employ MLK, which may require complex mathematical models and perfect CSI to achieve high detection performance~\cite{guo2018noncoherent, guo2018exploiting}. 
	Thus, this approach introduces significant complexity and dependency on accurate CSI.
	To overcome these limitations, we develop a low-complexity {Deep Learning (DL)}-based detector that can effectively detect the backscattered signals. 
	The rationale behind is that DL has the ability to learn directly from data (e.g., received signals), eliminating the need of complex mathematical models and perfect CSI.
	Note that in the literature, several works consider using DL for signal detection and channel estimation~\cite{ye2018power, fan2019cnn, rajendran2018deep, oshea2018over, van2022defeating}.
	However, most of them consider conventional signals (e.g., OFDM, BPSK, and QAM64 signals), and thus it may {not perform well} for very weak signals like AmB signals.
	In particular, the authors in~\cite{ye2018power, fan2019cnn} consider DL-based detectors for OFDM signals.
	Whereas, the studies in~\cite{rajendran2018deep} and \cite{oshea2018over} consider the signal classification task, i.e., predicting the modulation type of signals.
	Unlike the above works, in~\cite{van2022defeating}, the authors propose an AmB signal detector based on the Long Short-Term Memory (LSTM) architecture.
	Since LSTM requires more computing {capability} than that of the conventional architecture, e.g., fully connected architecture~\cite{freire2022computational}, it does not fit well in our considered lightweight anti-eavesdropping framework.
	Therefore, this work proposes a new  detector with low complexity elements, such as {$tanh$} activation, and a few small-size fully connected hidden layers.
	{By doing so, our DL-based detector offers a more practical and efficient solution for backscatter signal detection at the receiver.}
	
	Although DL-based detectors can achieve good detection performance, they usually require a large amount of high-quality data {(i.e., a collection of received signals)} for the training process~\cite{nichol2018first}.	
	This makes DL less efficient in practice when data is expensive and/or contains noise due to the {wireless} environment's dynamics and uncertainty.
	For example, new objects {(e.g. {the passing of a} bus)} can significantly impact wireless channel conditions, even changing links from line-of-sight (LOS) to non-line-of-sight (NLOS).
	As a result, in conventional DL, models may need to be retrained from scratch with newly collected data, and thus it is a time-consuming task~\cite{hospedales2021meta}.
	In this context, meta-learning (i.e., {learning how to learn}) emerges as a promising approach to quickly learn a new task with limited training data~\cite{finn2017model}. 
	Given the above, this work develops a meta-learning algorithm to train the DL model to quickly achieve a good detection performance in new environments. 
	Extensive simulation results substantiate the effectiveness of our proposed solution in effectively mitigating eavesdropping attacks. 
	{They also show that the proposed DL-based signal detector, without requiring perfect CSI, can attain a comparable Bit Error Ratio (BER) to the MLK-based detector, which is an optimal detector requiring a complex mathematical model and perfect CSI.}
	The main contributions of this work are:
	\begin{itemize}
		\item Propose a novel anti-eavesdropping framework leveraging the AmB {communications}. 
		In particular, we propose to use a low-cost and low-complexity AmB tag to assist in transmitting a part of information to the receiver by backscattering right on the transmit signals. 
		This solution is expected to open a new direction for future anti-{eavesdropping} communications. 
		\item Develop the DL-based detector to detect the AmB signals {at} the receiver to {overcome limitations of the} conventional MLK-based approaches.
		In particular, we propose a low complexity data preprocessing and design a lightweight Deep Neural Network (DNN) architecture to effectively detect AmB signals.
		\item Develop a meta-learning algorithm to quickly achieve {high} detection performance in new environments with little knowledge.
		The main idea of meta-learning is to utilize knowledge obtained from similar environments in order to {reduce} the size of the necessary training dataset, while still preserving the quality of learning. 
		\item Perform {extensive} simulations in several scenarios to get insights into various aspects of our proposed framework, such as maximum achievable rate, security, and robustness. 
		We also analyze the security level of our proposed framework based on the guessing entropy for the worse case when the eavesdropper has some prior information about the AmB communication settings.
	\end{itemize}
	
	The remaining of this paper is organized as follows.
	Our proposed anti-eavesdropping system and the channel model are discussed in Sections~\ref{Sec:system} and \ref{Sec.signalmodel}.
	Then, Sections IV and V present the MLK-based detector and our proposed DL-based detector for the AmB signal, respectively.
	Section VI discusses our simulation results.
	Finally, Section VII wraps up our paper with a conclusion.	
	
	\section{System Model}
	\label{Sec:system}
	\begin{figure}[t]
		\centering
		\includegraphics[width=0.50\linewidth]{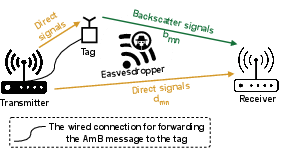}
		\caption{Anti-eavesdropping attack System model.}
		\label{Fig.system_model}
		\vspace{-20pt}
	\end{figure}
	
	This work considers a wireless network with the presence of an eavesdropper, as shown in Fig.~\ref{Fig.system_model}.
	Here, the transmitter has a single antenna, while the receiver is equipped with $M$ antennas. 
	The eavesdropper aims to wiretap the transmitted signals to {gather} the information in this channel.
	To cope with the eavesdropper, this work deploys a low-cost and low-complexity tag equipped with {the} AmB technology, allowing it to send data by backscattering ambient RF signals without producing active signals as in conventional active transmissions.
	Specifically, the AmB tag has two operation states, including (i) the absorbing state, where the tag does not reflect incoming signals, and (ii) the reflecting state, where the tag reflects incoming signals.
	In this way, the AmB tag can transmit data without using any active RF components.
	Before sending a message to the receiver, the transmitter first encodes it into two messages: (i) an active message for the transmitter and (ii) an AmB message forwarded to the AmB tag via the wired channel. 

	The active message is then directly transmitted by the transmitter to the receiver using {the conventional RF transmission method.}
	At the same time, when the transmitter transmits the active message, the AmB tag will leverage its RF signals to {backscatter} the AmB message~\cite{hoang2020ambient, huynh2018survey}. 
	Thus, the receiver will receive two data streams simultaneously, one over the conventional channel $d_{mn}$ and another over the backscatter channel $b_{mn}$, as depicted in Fig.~\ref{Fig.system_model}.
	Note that instead of producing active signals, the AmB tag only backscatters/reflects signals.
	Therefore, wiretapping AmB signals is intractable unless the eavesdropper {has prior knowledge about} the system configuration, i.e., the usage of AmB and the exact backscatter rate.
	As a result, our system can provide a new deception strategy for data transmissions.
	Specifically, the eavesdropper is attracted by the active signals generated by the transmitter, so it pays less attention to (or {is} even unaware of) the presence of AmB communications.
	{Here, we assume} that without obtaining information from the AmB message, the eavesdropper cannot decode the information from the original message. 
	Moreover, even in the worse case in which the eavesdropper knows the presence of AmB communication in the system, they cannot obtain the original message straightforwardly without {the} knowledge about the system's message-splitting mechanism.
	As a result, our proposed approach can deal with eavesdropping attacks in wireless systems.
	\begin{figure}[t]
		\centering
		\includegraphics[width=0.50\linewidth]{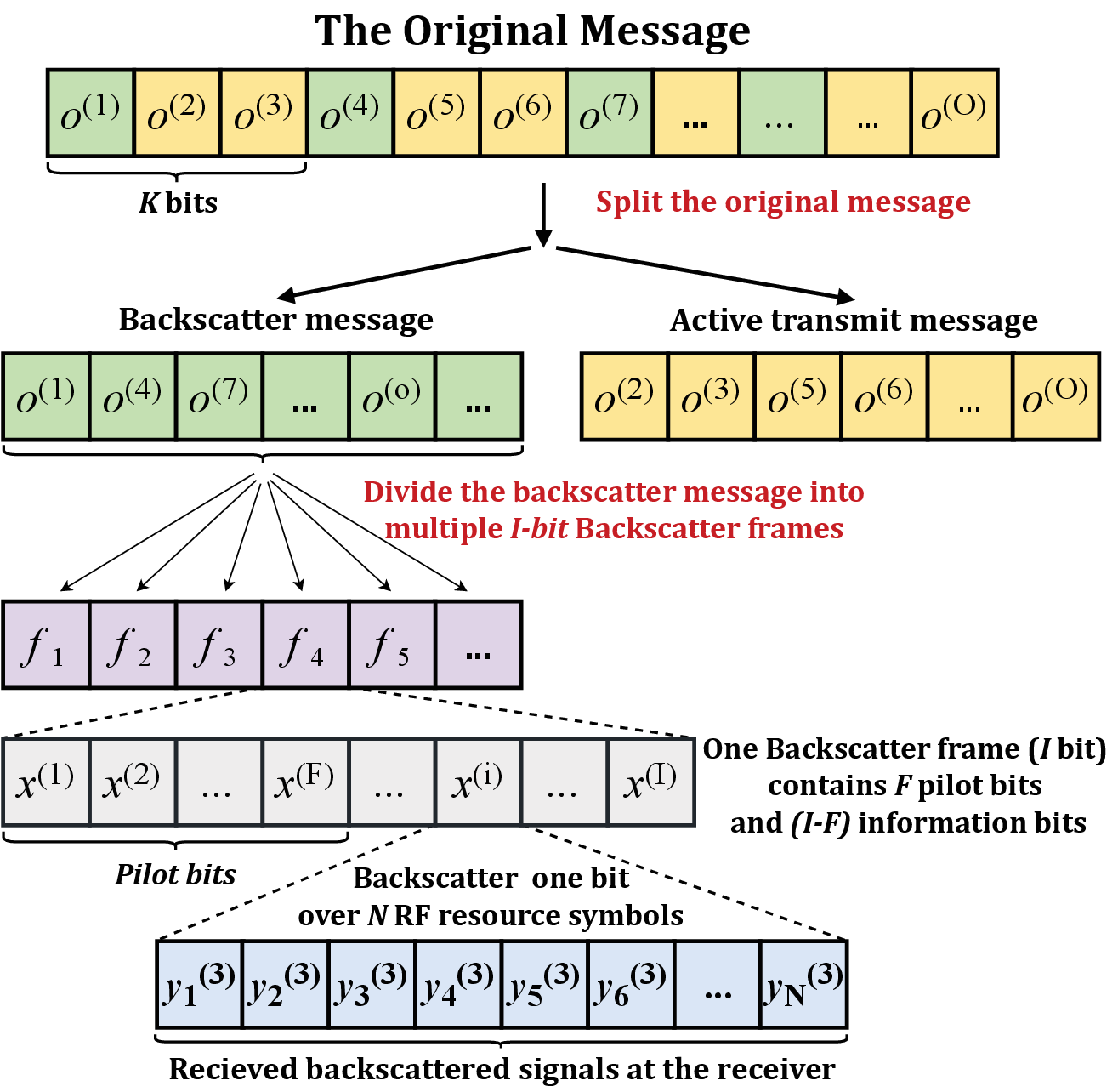}
		\caption{Message-splitting mechanism.}
		\label{Fig.splitmessage}
		\vspace{-20pt}
	\end{figure}
	
	The detail of our proposed encoding mechanism is depicted in Fig.~\ref{Fig.splitmessage}.
	Specifically, at the beginning, an encoding mechanism is used to split the original message into two parts, i.e., backscatter message and active transmit message. 	
	Note that {our framework here and following analysis can adopt} any encoding mechanism, and the design of encoding mechanism is out of {the} scope of this work. 
	Since the AmB rate is {usually} lower than the active transmission rate, the AmB message's size can be designed to be smaller than the active message's size.
	As such, the AmB message is constructed by taking bits at every $K$ bits of the original message.
	By doing so, the system security under the presence of the eavesdropper can be significantly improved because the eavesdropper is unable to derive the message splitting mechanism and the backscatter message.
	{Note that to improve the detection performance, we incorporate $F$ pilot bits into the AmB frame, as depicted in Fig.~\ref{Fig.splitmessage}.
		The specific utilization and significance of these pilot bits will be extensively discussed in detail in Section~\ref{subsec:data_procesing}.}

	\section{Channel Model}
	\label{Sec.signalmodel}
	This section presents the channel model of our considered system in detail. 
	In particular, the AmB rate should be significant lower than the sampling rate of the transmitter's signals so that the receiver can decode the backscattered signals with low BERs~\cite{huynh2018survey, liu2013ambient, van2022defeating}.
	Formally, the transmit signals' sampling rate is assumed to be $N$ times higher than the AmB rate.
	In other words, each bit of the backscatter message is backscattered over $N$ transmitter symbols.
	In the considered system, the receiver has $M$ antennas {(\mbox{$M\geq1$})}. 
	Let $y_{mn}$ denote the signal that the $m$-th antenna of the receiver receives at the $n$-th time.
	As illustrated in Fig.~\ref{Fig.system_model}, $y_{mn}$ comprises (i) the {active signal transmitted} on the direct link, (ii) the signals backscattered by the AmB tag on the backscatter link, and (iii) noise from the surrounding environment. 
	As such, $y_{mn}$ can be expressed as follows:
	\begin{equation}
		\label{eq:totalreceive}
		y_{mn} = \underbrace{d_{mn}}_{\text{direct link}} + \underbrace{b_{mn}}_{\text{backscatter link}} + ~\sigma_{mn},
	\end{equation}
	where 
	$\sigma_{mn}$ is the noise following the unit-variance and zero-mean circularly symmetric complex Gaussian (CSCG), denoted by \mbox{$\sigma_{mn} \sim \mathcal{CN}(0,1)$}. 
	In the following, the signals on the direct and backscatter links are formally formulated.
	
	\subsection{Direct Channel}
	On the direct channel, i.e., the transmitter-receiver link, the average power received by the receiver is formulated by
	\begin{equation}
		P_{tr} = \frac{ \kappa P_t G_t G_r}{ {L_r}^{\upsilon}},
	\end{equation}
	where $P_t$ is the transmitter's transmit power, and $\kappa = {(\frac{\lambda}{4\pi})}^2$ {in which} $\lambda$ is the wavelength,  $\upsilon$ denotes the path loss exponent.
	$L_r$ is the transmitter-receiver distance. 
	The antenna gains of the transmitter and the receiver are denoted by $G_t$ and $G_r$, respectively. 
	At time instant $n$, we denote the signal transmitted by the transmitter as $s_{tn}$. 
	Then, at the receiver's $m$-th antenna, the direct link signal is given by
	\begin{equation}
		\label{eq:direct}
		d_{mn} = f_{dm} \sqrt{P_{tr}}s_{tn}, 
	\end{equation}
	where $f_{dm}$ represents the Rayleigh fading such that $\mathbb{E}[|f_{dm} |^2] = 1$~\cite{van2022defeating}. Note that the proposed AmB tag uses the AmB communication technology to backscatter $s_{tn}$ without using any dedicated energy sources. 
	Since the AmB tag does not know the active message, $s_{tn}$ appears as random signals~\cite{guo2018noncoherent,zhang2018constellation, van2022defeating}. Therefore, we can assume that \mbox{$s_{tn} \sim \mathcal{CN}(0,1)$}.
	
	\subsection{Ambient Backscatter Channel}
	As mentioned, the AmB tag backscatters the transmitter's active signals to transmit AmB messages. 
	In the following, we present the formal formulation of the AmB signals at the receiver. 
	Firstly, the average power received by the AmB tag can be defined by
	\begin{equation}
		P_{b} = \frac{\kappa P_t G_t G_b}{{L_b}^{\upsilon}},
	\end{equation}
	where $G_b$ denotes the antenna gain of the AmB tag, and $L_b$ denotes the transmitter-tag distance. 
	Let $g_r$ denote the Rayleigh fading of the link from the transmitter to the AmB tag, 
	then the active signals from the transmitter at the AmB tag is given by
	\begin{equation}
		l_n = \sqrt{P_{b}}g_r s_{tn}.
	\end{equation}
	As discussed above, the key idea of the AmB communication is to absorb or reflect surrounding ambient RF signals to convey information without generating any active signals. 
	As such, we denote $e$ as the state of the AmB tag.
	Specifically, \mbox{$e=1$} when the tag reflects the transmitter's signals, i.e., transmitting bits 1, and \mbox{$e=0$} when the tag absorbs the transmitter's signals, i.e., transmitting bits 0. 
	Because each {backscattered} bit is backscattered over $N$ transmitter's symbols, state $e$ of the AmB tag remains unchanged during this period. 
	The backscattered signals then can be given by
	\begin{equation}
		s_{bn} = \gamma l_n e,
	\end{equation}
	where $\gamma$ is the reflection coefficient. 
	It is worth noting that $\gamma$ captures all properties of the AmB tag such as load impedance and antenna impedance~\cite{zhang2018constellation}. 
	Let $f_{bm}$ represent the Rayleigh fading of the AmB channel, and $L_e$ denotes the AmB tag-to-receiver distance. 
	Here, we can assume that $\mathbb{E}[|f_{bm}|^2] = 1$ and $\mathbb{E}[|g_r|^2] = 1$ without loss of generality~\cite{van2022defeating}. 
	Then, the signals in the AmB link received by the $m$-th antenna is given by
	\begin{equation}
		\label{eq:bn}
		\begin{aligned}
			b_{mn} &= f_{bm} \sqrt{\frac{G_b G_r \kappa}{{L_e}^{\upsilon}}} s_{bn}\\
			&= f_{bm} \sqrt{\frac{G_b G_r \kappa}{{L_e}^{\upsilon}}} \gamma e \Big(g_r \sqrt{P_{b}}s_{tn}\Big)\\
			&= f_{bm} e \Big(g_r \sqrt{\frac{\kappa |\gamma|^2 G_b G_r P_b}{{L_e}^{\upsilon}}} s_{tn}\Big)\\
			&= f_{bm} e \Big(g_r \sqrt{\frac{\kappa |\gamma|^2 G_b G_r \kappa P_t G_t G_b}{{L_e}^{\upsilon} {L_b}^{\upsilon}}}s_{tn}\Big)\\
			&= f_{bm} e \Big(g_r \sqrt{\frac{\kappa |\gamma|^2 P_{tr}{G_b}^2{L_r}^{\upsilon}}{{L_b}^{\upsilon}{L_e}^{\upsilon}}} s_{tn}\Big).
		\end{aligned}
	\end{equation}	
	By denoting $\tilde{\alpha}_r = \frac{\kappa {|\gamma|^2}{G_b}^2{L_r}^{\upsilon}} {{L_b}^{\upsilon} {L_e}^{\upsilon}}$, (\ref{eq:bn}) can be rewritten as 
	\begin{equation}
		\label{eq:backscatter}
		b_{mn} = f_{bm} e \Big(g_r \sqrt{\tilde{\alpha}_r P_{tr}}s_{tn} \Big).
	\end{equation}
	
	\subsection{Received Signals}
	\label{sec:received}
	Now, we can obtain the received signals at the receiver's $m$-th antenna by substituting (\ref{eq:direct}) and (\ref{eq:backscatter}) to (\ref{eq:totalreceive}) as follows:
	\begin{equation}
		\label{eq:total_received_final}
		\begin{aligned}
			y_{mn} &= \underbrace{d_{mn}}_{\text{direct link}} + \underbrace{b_{mn}}_{\text{backscatter link}} + \sigma_{mn},\\
			&=f_{dm} \sqrt{P_{tr}}s_{tn} + f_{bm} e \Big(g_r \sqrt{\tilde{\alpha}_r P_{tr}}s_{tn} \Big) + \sigma_{mn}.
		\end{aligned}
	\end{equation}
	Let $\alpha_{dt} \triangleq P_{tr}$ and $\alpha_{bt} \triangleq \tilde{\alpha}_r P_{tr}$ denote the average signal-to-noise ratios (SNRs) of the transmitter-to-receiver link and the transmitter-tag-receiver channel (i.e., backscatter link), respectively. 
	Hence, \eqref{eq:total_received_final} is rewritten as follows:
	\begin{equation}
		y_{mn} = \underbrace{f_{dm} \sqrt{\alpha_{dt}}s_{tn}}_{\text{direct link}} + \underbrace{f_{bm}  e \Big(g_r \sqrt{\alpha_{bt}}s_{tn} \Big)}_{\text{backscatter link}} + \sigma_{mn}.
	\end{equation}
	
	Since the receiver has $M$ antennas, the channel response vectors corresponding to the backscatter and the direct channels and are respectively given by
	\begin{equation}
		\begin{aligned}
			&\mathbf{f}_{b} = [f_{b1}, \ldots, f_{bm}, \ldots, f_{bM}]^\mathsf{T},\\
			&\mathbf{f}_{d} = [f_{d1}, \ldots, f_{dm}, \ldots, f_{dM}]^\mathsf{T}.
		\end{aligned}
	\end{equation}
	Then, the receiver's total received signals is expressed as
	\begin{equation}
		\begin{aligned}
			\mathbf{y}_{n} 
			&= \underbrace{\mathbf{f}_{d} \sqrt{\alpha_{dt}}s_{tn}}_{\text{direct link}} + \underbrace{\mathbf{f}_{b}  e \Big(g_r \sqrt{\alpha_{bt}}s_{tn} \Big)}_{\text{backscatter link}} + \bm{\sigma}_{n},
		\end{aligned}
	\end{equation}
	where $\bm{\sigma}_{n} = [\sigma_{1n}, \ldots, \sigma_{mn}, \ldots, \sigma_{Mn}]^\mathsf{T}$.
	
	This work {considers} that there are $I$ bits in each backscatter frame, i.e., 
	\begin{equation}
		\mathbf{x} = [x^{(1)},\ldots, x^{(i)}, \ldots, x^{(I)}].
	\end{equation}
	To enhance the receiver's decoding process, $F$ pilot bits are placed in each backscatter frame, so there are $I-F$ information bits in each backscatter frame. 
	We consider that the AmB tag and the receiver know these bits in advance and use them to estimate the AmB tag-receiver channel coefficients.
	More details about the use of pilots will be discussed in Section \ref{sec:deeplearning}. 
	We assume that during one AmB frame, the AmB tag-receiver channel is invariant~\cite{guo2018noncoherent}. 
	Since each backscatter bit is backscattered over $N$ symbols transmitted by the transmitter, we can express the receiver's received signals during the $i$-th AmB symbol duration by:
	\begin{equation}
		\mathbf{y}_{n}^{(i)} = \mathbf{f}_{d} \sqrt{\alpha_{dt}}s_{tn}^{(i)} + \mathbf{f}_{b}  e^{(i)} \Big(g_r \sqrt{\alpha_{bt}}s_{tn}^{(i)} \Big) + \bm{\sigma}_{n}^{(i)},
	\end{equation}
	where the backscatter state $e^{(i)}$ equals backscatter bit $x^{(i)}$, i.e., \mbox{$e^{(i)} = x^{(i)}, \forall i = 1,2, \ldots, I$}, and \mbox{$n = 1,2, \ldots, N$}.
	
	Note that for the AmB tag, the transmitter's signals appear as random signals and are unknown in advance. 
	Therefore, the AmB rate's {closed}-form (denoted $R_\text{AmB}$) cannot be obtained~\cite{huynh2018survey,guo2018noncoherent}. 
	For that, in Theorem~\ref{theo:maxrate}, we provide an alternative method to derive the AmB system's idealized throughput model when $N=1$.
	Let $\theta_0$ and $\theta_1$ denote the probability of backscattering bit 0 and bit 1, respectively. 		
		$\mathcal{V}$ denotes a realization of $\mathbf{y}$, and $p(\mathcal{V}|e)$ {denotes} the conditional probability density function.
		Theorem 1 is as follows.
	\begin{theorem}
		\label{theo:maxrate}		
		The maximum achievable rate of the AmB tag is given as
		\begin{equation}
			\label{eq.maxrate}
			\begin{aligned}
				R_\text{AmB}^* = Z(\theta_0) - \int_{\mathcal{V}}\big(\theta_1 p(\mathcal{V}|e=1) + \theta_0p(\mathcal{V}|e=0)  \big)Z(\mu_0)d\mathcal{V},
			\end{aligned}			
		\end{equation}		
		where \mbox{$\mu_0 = p(\mathcal{V}|e=0)$}, and $Z(\cdot)$ represents the binary entropy function.
	\end{theorem}
	\begin{proof}
		The proof of Theorem~\ref{theo:maxrate} is provided in Appendix~\ref{appendix:maximumbackscatterrate}.
	\end{proof}
	
	It can be observed from~\eqref{eq.maxrate} that the maximum achievable rate $R_\text{AmB}^*$ depends on the probability of backscattering bit 0 $\theta_0$.
	We will explore this dependency through our simulations in Section~\ref{sec:evaluation}.
	Let $\mathbf{Y}^{(i)} \!=\! [\mathbf{y}_1^{(i)}, \ldots, \mathbf{y}_N^{(i)}]^\mathsf{T}$ denote the sequence of received signal corresponding to the period of the $i$-th AmB symbol. 
	We describe our proposed AmB detectors that use MLK and DL to recover the original bits $x^{(i)}$ from these sequences of received signals in the following sections.
	
	\section{AmB Signal Detector based on Maximum Likelihood}
	\label{sec:decoding}
	Recall that signals backscattered by the AmB tag can be regarded as the active signals' background noise, making them very challenging to be detected~\cite{van2018reinforcement, zhang2016secure}.
	This section presents the AmB signal decoding based on MLK. 
	Note the MLK-based detector can be considered as an optimal signal detector, so it {can provide} the system's upper-bound {performance for comparison purpose}.
	
	\subsection{Received Signals' Likelihood Functions}
	
	The aim of MLK-based detector is to derive the received signals' likelihood functions. 
	In particular, if the AmB tag transmits bits ``0'' (corresponding to \mbox{$e^{(i)} = 0$}), the receiver only receives signals solely from the transmitter-to-receiver channel.
	Whereas, if the AmB tag transmits bits ``1'' (corresponding to \mbox{$e^{(i)} = 1$}), the receiver will receive signals from both direct and AmB channels. 
	Consequently, the channel statistical covariance matrices for these scenarios are expressed as ~\cite{guo2018noncoherent},~\cite{zhang2018constellation}
	
	\begin{equation}
		\begin{aligned}
			\mathbf{R}_1 &= (\mathbf{k}_1+\mathbf{k}_2)(\mathbf{k}_1+\mathbf{k}_2)^\text{H}+ \mathbf{I}_M,~\text{if $e^{(i)} = 1$},\\
			\mathbf{R}_0 &= \mathbf{k}_1 \mathbf{k}_1^\text{H} + \mathbf{I}_M,~\text{if $e^{(i)} = 0$},
		\end{aligned}
	\end{equation} 
	where $\mathbf{k}_1 = \mathbf{f}_d \sqrt{\alpha_{dt}}$, $\mathbf{k}_2 =  g_r \mathbf{f}_b \sqrt{\alpha_{bt}}$, $(\cdot)^\text{H}$ is the conjugate transpose operator, and $\mathbf{I}_M$ is the identity matrix with size \mbox{$M \times M$}. 
	It is important to note that the estimation of CSI can be performed by various techniques which are well studied in the literature~\cite{ma2018blind, zhao2019channel}. 
	Similar to~\cite{guo2018noncoherent}, the noise and the transmitter's RF signals are assumed to follow the CSCG distribution\footnote{Note that this information is not required by our proposed DL-based AmB signal detector.}.
	As a {result}, $\mathbf{y}_n^{(i)}$ also follows the CSCG~\cite{guo2018exploiting}, i.e., 	
	\begin{equation}
		\mathbf{y}_{n}^{(i)} \sim \left\{
		\begin{array}{ll}						
			\mathcal{CN}\left(0, \mathbf{R}_0\right), &\mbox{if $e^{(i)}=0$},\\
			\mathcal{CN}\left(0, \mathbf{R}_1\right), &\mbox{if $e^{(i)}=1$}.												
		\end{array}	\right.
	\end{equation}
	Now, the conditional probability density function (PDFs) of the received signals $\mathbf{y}_n^{(i)}$ given the backscatter state $e^{(i)}$ is computed by
	\begin{equation}
		\label{eq:pdf}
		\begin{aligned}
			p(\mathbf{y}_n^{(i)}|e^{(i)} = 0) &= \frac{1}{\pi^M |\mathbf{R}_0|} e^{-{\mathbf{y}_n^{(i)}}^\text{H} \mathbf{R}_0^{-1}\mathbf{y}_n^{(i)}}, \\
			p(\mathbf{y}_n^{(i)}|e^{(i)} = 1) &= \frac{1}{\pi^M |\mathbf{R}_1|} e^{-{\mathbf{y}_n^{(i)}}^\text{H} \mathbf{R}_1^{-1}\mathbf{y}_n^{(i)}},
		\end{aligned}
	\end{equation}
	where $|\cdot|$ and $(\cdot)^{-1}$ represent the determinant and inverse of a matrix, respectively.
	From~(\ref{eq:pdf}), we now can express the likelihood functions $\mathcal{H}(\cdot)$ for the sequence of received signals as follows~\cite{guo2018noncoherent, van2022defeating}:
	\begin{equation}
		\label{eq:likelihoodfunctions}
		\begin{aligned}
			\mathcal{H}(\mathbf{Y}^{(i)}|e^{(i)} = 1) &= \prod_{n=1}^{N} \frac{1}{\pi^M |\mathbf{R}_1|} e^{-{\mathbf{y}_n^{(i)}}^\text{H} \mathbf{R}_1^{-1}\mathbf{y}_n^{(i)}},\\
			\mathcal{H}(\mathbf{Y}^{(i)}|e^{(i)} = 0) &= \prod_{n=1}^{N} \frac{1}{\pi^M |\mathbf{R}_0|} e^{-{\mathbf{y}_n^{(i)}}^\text{H} \mathbf{R}_0^{-1}\mathbf{y}_n^{(i)}}.
		\end{aligned}
	\end{equation}
	
	\subsection{Maximum Likelihood Detector}
	Let $\tilde{e}^{(i)}$ be the estimation of $e^{(i)}$.
	From (\ref{eq:likelihoodfunctions}), the MLK hypothesis in~\eqref{eq:likelihood_criterion} can be used to obtain the backscattered state $e^{(i)}$~\cite{guo2018noncoherent},~\cite{zhang2018constellation}.
	\begin{equation}
		\label{eq:likelihood_criterion}
		\tilde{e}^{(i)} 	=	\left\{	\begin{array}{ll}
			1,	&	\mathcal{H}(\mathbf{Y}^{(i)}|e^{(i)} = 0) < \mathcal{H}(\mathbf{Y}^{(i)}|e^{(i)} = 1),\\
			0,	&	\mathcal{H}(\mathbf{Y}^{(i)}|e^{(i)} = 0) > \mathcal{H}(\mathbf{Y}^{(i)}|e^{(i)} = 1).			
		\end{array}	\right.
	\end{equation}
	Based on~\eqref{eq:likelihoodfunctions} and \eqref{eq:likelihood_criterion}, the MLK hypothesis can be given by
	\begin{equation}
		\label{eq:likelihood_criterion_2}
		\tilde{e}^{(i)} 	\!=\!	\left\{	\begin{array}{ll}
			\!\!\!1,	\prod_{n=1}^{N} \!\frac{1}{\pi^M |\mathbf{R}_0|} e^{-{\mathbf{y}_n^{(i)}}^\text{H} \mathbf{R}_0^{-1}\mathbf{y}_n^{(i)}} \!\!<\!\! \prod_{n=1}^{N} \!\frac{1}{\pi^M |\mathbf{R}_1|} e^{-{\mathbf{y}_n^{(i)}}^\text{H} \mathbf{R}_1^{-1}\mathbf{y}_n^{(i)}},\\
			\!\!\!0,	\prod_{n=1}^{N} \!\frac{1}{\pi^M |\mathbf{R}_0|} e^{-{\mathbf{y}_n^{(i)}}^\text{H} \mathbf{R}_0^{-1}\mathbf{y}_n^{(i)}} \!\!>\!\! \prod_{n=1}^{N} \!\frac{1}{\pi^M |\mathbf{R}_1|} e^{-{\mathbf{y}_n^{(i)}}^\text{H} \mathbf{R}_1^{-1}\mathbf{y}_n^{(i)}}.
			
		\end{array}	\right.
	\end{equation}	
	By applying logarithmic operations to \eqref{eq:likelihood_criterion_2}, we can derive the following expression
	\begin{equation}
		\tilde{e}^{(i)}	=	\left\{	\begin{array}{ll}
			1,	&	\sum_{n=1}^{N} {\mathbf{y}_n^{(i)}}^\text{H} (\mathbf{R}_0^{-1}-\mathbf{R}_1^{-1})\mathbf{y}_n^{(i)} > N\ln\frac{|\mathbf{R}_1|}{|\mathbf{R}_0|},\\
			0,	&	\sum_{n=1}^{N} {\mathbf{y}_n^{(i)}}^\text{H} (\mathbf{R}_0^{-1}-\mathbf{R}_1^{-1})\mathbf{y}_n^{(i)} < N\ln\frac{|\mathbf{R}_1|}{|\mathbf{R}_0|}.
		\end{array}	\right.
	\end{equation}
	Then, the backscattered bit $x^{(i)}$ {can be} derived based on $\tilde{e}^{(i)}$. 
	However, for that, MLK requires perfect CSI for both conventional and AmB transmissions, making its less efficient in practice, especially for AmB communications where backscattered signals are very weak~\cite{guo2018noncoherent}.
	
	\section{Deep Learning-Based Signal Detector}
	\label{sec:deeplearning}
	As discussed in Section~\ref{sec:decoding}, the traditional approaches for backscattered signal detection (i.e., MLK-based detectors) require accurate channel statistical covariance matrices. 
	However, the uncertainty of the wireless environment makes it difficult for the system to estimate these matrices accurately. 
	To that end, we develop a DL-based signal detector that can overcome this challenge.
	We first discuss the data preprocessing procedure, aiming to construct a dataset for the training phase of the DL detector, and then our proposed {DNN} architecture is {presented} in detail.
	
	\subsection{Data Preprocessing} 
	\label{subsec:data_procesing}
	In practice, the raw data (e.g., received signals) may not be qualified to be used directly to train a DL model {as it can lead to a poor or even useless trained model}.
	Therefore, data preprocessing is an essential step before training.
	In this work, the data preprocessing steps are as follows.
	Recall that $\mathbf{y}^{(i)}_n$ denotes the received AmB signals, and $\mathbf{Y}^{(i)}$ is the sequence of $\mathbf{y}^{(i)}_n$ corresponding to the $i$-th AmB symbol period. 
	Instead of directly using $\mathbf{Y}^{(i)}$ for model training, the sample covariance matrix ${\mathbf{S}^{(i)}}$ of the received signals, defined by~\eqref{eq:sampleCV}, can be used as the training data for DL models.
	\begin{equation}
		\label{eq:sampleCV}
		{\mathbf{S}}^{(i)} = \frac{1}{N} \sum_{n=1}^{N}\mathbf{y}^{(i)}_n{\mathbf{y}^{(i)}_n}^\mathrm{H}.
	\end{equation}
	The main reason is that the sample covariance matrix can capture the relationship between received signals corresponding to an AmB symbol period
	, thus providing more insights on received signals than {the signals} themselves.
	In addition, the performance of a signal detector is proportional to the number of RF symbols corresponding to an AmB symbol, i.e., $N$.
	In other words, the higher value $N$ is, the better performance the detector can achieve.
	Typically, $N$ is often set to a higher value than the number of antennas~\cite{guo2018exploiting, guo2018noncoherent, van2022defeating}, and thus the size of ${\mathbf{S}^{(i)}}$ is {smaller} than that of $\mathbf{Y}^{(i)}$.
	Therefore, training the DNN model with the sample covariance matrix ${\mathbf{S}^{(i)}}$ can not only preserve all crucial information of the received signals but also reduce the training data size without affecting the learning efficiency~\cite{van2022defeating, guo2018noncoherent}. 
	\begin{figure*}[t]
		\centering
		\includegraphics[width=0.7\linewidth]{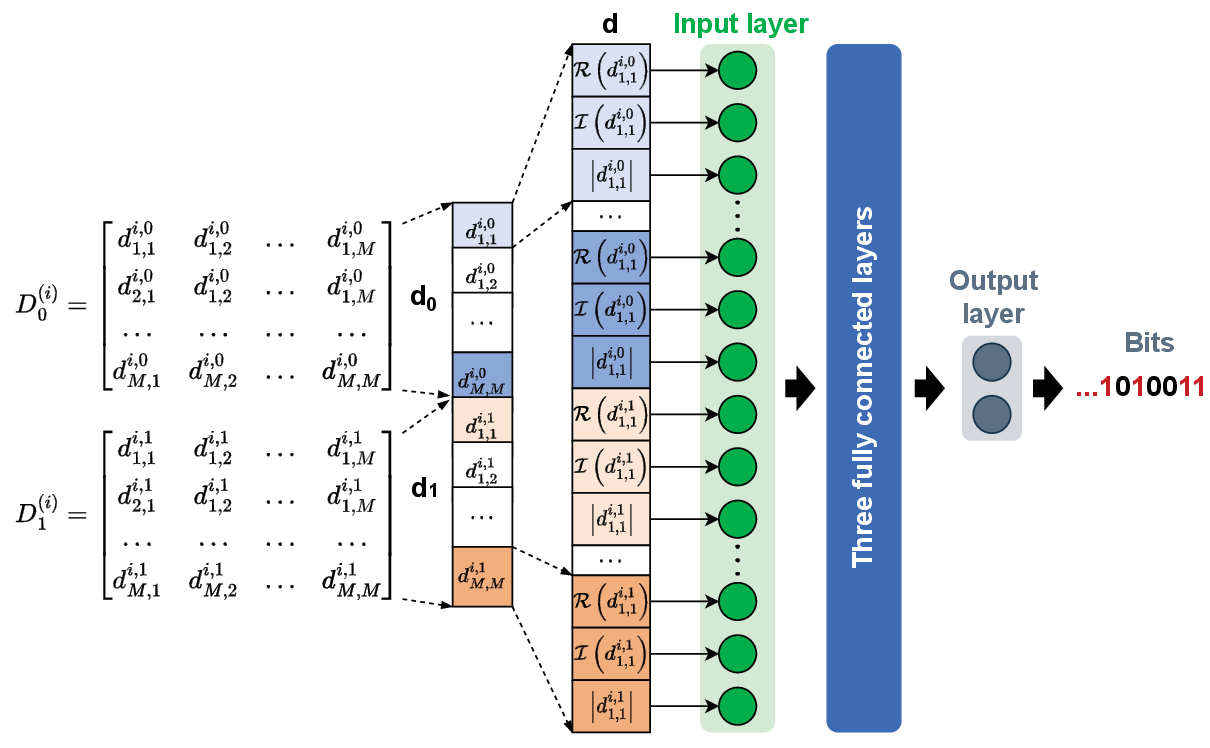}
		\caption{Our proposed AmB signal detector based on DL.}
		\label{fig:DL_detector}
		\vspace{-20pt}
	\end{figure*}
	Note that conventional MLK-based detectors explicitly require accurate estimation of channel coefficients (i.e., $\mathbf{k_1}$ and $\mathbf{k_2}$), resulting in high computational complexity, especially for a system with multiple antennas~\cite{shen2006channel}.
	To alleviate this challenge, pilot bits are leveraged to indirectly capture the channel coefficients, thus improving the accuracy of the covariance matrix estimation.
	Specifically, we use a sequence of $F$ pilot bits where the first $F/2$ bits are ``0'', and the rest are ``1''.
	The detector knows these bits in advance so that the channel coefficients can be implicitly obtained by 
	inspecting the received signals corresponding to them. 
	In this work, the averages of ${\mathbf{S}}^{(i)}$ of received signals corresponding to pilot bits ``1'' and ``0'' are leveraged to indirectly capture the channel coefficients, which are given as follows:
	\begin{equation}
		\begin{aligned}
			\bar{\mathbf{R}}_1 &= \frac{2}{FN} \sum_{i=1}^{F/2} \sum_{n=1}^{N} \mathbf{y}^{(i)}_n{\mathbf{y}^{(i)}_n}^\mathrm{H}, \quad \textrm{when}~e^{(i)}=1,\\
			\bar{\mathbf{R}}_0 &= \frac{2}{FN} \sum_{i=1}^{F/2} \sum_{n=1}^{N} \mathbf{y}^{(i)}_n{\mathbf{y}^{(i)}_n}^\mathrm{H}, \quad \textrm{when}~e^{(i)}=0.				
		\end{aligned}		
	\end{equation}
	
	In this way, the detector can use $\bar{\mathbf{R}}_0$ and $\bar{\mathbf{R}}_1$ to effectively detect AmB signals without requiring channel coefficients explicitly.
	In particular, if the AmB tag backscatters bit ``0'', the sample covariance matrix $\mathbf{S}^{(i)}$ would be similar to $\bar{\mathbf{R}}_0$ but different from $\bar{\mathbf{R}}_1$.
	In contrast, if bit ``1'' is backscattered, the sample covariance matrix $\mathbf{S}^{(i)}$ would be similar to $\bar{\mathbf{R}}_1$ but different from $\bar{\mathbf{R}}_0$.	
	It is worth noting that these similarities and differences between $\mathbf{S}^{(i)}$ and the averaged covariance matrices (i.e., $\bar{\mathbf{R}}_0$ and $\bar{\mathbf{R}}_1$) can be better represented by multiplying $\mathbf{S}^{(i)}$ with the inverse of these matrices.	
	As such, the data for training our DL model is given as follows:
	\begin{equation}
		\begin{aligned}
			\mathbf{D}^{(i)}_1 &= \mathbf{S}^{(i)}{\bar{\mathbf{R}}_1}^{-1},\\			
			\mathbf{D}^{(i)}_0 &= \mathbf{S}^{(i)}{\bar{\mathbf{R}}_0}^{-1}.
		\end{aligned}		
	\end{equation}
	By doing so, the proposed AmB signal detector based on DL can obtain adequate information about the channel coefficients via $\mathbf{S}^{(i)}$, $\bar{\mathbf{R}}_0$, and $\bar{\mathbf{R}}_1$ to learn and detect backscattered bits.
	
	To create input data for the DL model, the matrices $\bar{\mathbf{R}}_0$ and $\bar{\mathbf{R}}_1$ are processed as follows.
	Recall that the receiver has $M$ antennas, $\mathbf{D}^{(i)}_0$ and $\mathbf{D}^{(i)}_1$ are $M\!\times\!M$ matrices.	
	Let $d^0_{ij}$ and $d^1_{ij}$ {be} the elements of $\mathbf{D}^{(i)}_0$ and $\mathbf{D}^{(i)}_1$, respectively, with \mbox{$i,j \in \{1,2, \dots,M\}$}.
	First, the matrices $\bar{\mathbf{R}}_0$ and $\bar{\mathbf{R}}_1$ are flattened into two vectors, i.e., $\mathbf{d}_0$ and $\mathbf{d}_1$, each has $M\!\times\!M$ elements, as shown in Fig.~\ref{fig:DL_detector}.
	Then, $\mathbf{d}_1$ is concatenated at the end of $\mathbf{d}_0$ to form vector $\mathbf{d}$ with $2M^2$ elements.
	Finally, the input data for the DL model is constructed by taking (i) the real part, (ii) the imaging part, and (iii) the absolute of each element of the combined vector $\mathbf{d}$, as illustrated in Fig.~\ref{fig:DL_detector}.
	{Thus}, the input data can allow the DL model to effectively leverage all important information of the received signals, thereby improving learning quality.
	
	\subsection{Deep Neural Network Architecture}
	\label{subsec:DL}
	This paper designs a DNN in the signal detector to decide whether the signals received by the receiver correspond to bit ``1'' or bit ``0'' of a AmB frame.
	Aiming to develop a lightweight anti-eavesdropping solution, we only consider a small footprint DNN with a simple and {typical} architecture. 
	In particular, our DNN consists of an input layer, {a small number (e.g., three)} of fully connected (FC) hidden layers, and an output layer, as illustrated in Fig.~\ref{fig:DL_detector}. 
	Since the input vector data's size is \mbox{$M\!\times\!M\!\times\!2\!\times3$}, the input layer consists of \mbox{$M\!\times\!M\!\times\!2\!\times3$} neurons.	
	After the input layer receives the input data, it will be sequentially forwarded to FC layers, and each consists of multiple neurons. 
	In FC layers, each neuron connects to all neurons in the previous layer.
	For example, each neuron in the first FC layer connects to all neurons in the input layer.	
	The output of the last FC layer is applied to the soft-max function at the output layer to obtain the category probabilities $p$ of which the received signals corresponds to bit ``0'' and bit ``1''. 
	Finally, the output layer uses these probabilities to determine the input's class, i.e., ``0'' or ``1''.
	In DNNs, a neuron applies a mathematical function, such as a sigmoid, tanh, or rectified linear unit (ReLU), to the sum of its input and bias to get its output.
	Here, the connections between neurons are represented by weight values. 
	The input to a neuron {is} calculated as the weighted sum of the outputs of neurons in the previous layer that are connected to it.
	The training process aims to optimize the DNN parameters $\theta$ (i.e., weights and bias) to minimize the loss $\mathcal{L}(f(\mathbf{d};\theta), \psi)$, which is the error between the DNN output $f(\mathbf{d};\theta)$ and the ground truth $\psi$, as follows:
	\begin{equation}
		\label{eq.loss}
		\min_\theta ~~\mathbb{E} \big[\mathcal{L}(f(\mathbf{d};\theta), \psi)\big],
	\end{equation}
	{here} the cross-entropy loss is used at the output layer.
	
	In DL, Stochastic Gradient Descent (SGD) is widely used to optimize the DNN parameters due to the simplicity in implementation while still achieving good performance~\cite{goodfellowdeep2016}.
	Specifically, at learning iteration $t$, a minibatch $\mathbf{B}_t$ is sampled from the training dataset.
	Then, the cost function is computed as follows:
	\begin{equation}
		J_t = \frac{1}{|\mathbf{B}_t|} \sum_{\mathbf{d}, \psi \in \mathbf{B}_t}\mathcal{L}\big(f(\mathbf{d};\theta_t), \psi \big).
	\end{equation}
	After that, the DNN's parameters are updated by
	\begin{equation} 	
		\theta_{t+1} \leftarrow \theta_t - \epsilon_t \nabla_{\theta_t} J_t(\theta_t),
	\end{equation}
	where $\epsilon_t$ is the step size controlling how much the parameters are updated, and $\nabla_{\theta_t}(\cdot)$ is the gradient operator with respect to the parameters $\theta$.	
	{Thanks to} this design, the proposed DL based-receiver can accurately predict transmitted bits in a backscatter frame even for weak AmB signals. 
	
	Although DL allows receivers to effectively detect AmB signals without requiring accurate {CSI}, it faces several shortcomings.
	First, the training process of DL often requires a huge amount of {received signals for the training process} to achieve good performance~\cite{nichol2018first}.	
	Thus, this makes DL less efficient in practice, especially when data is expensive and contains noise due to the environment's dynamic and uncertainty, as in the considered wireless environment.
	Second, conventional DL may face the over-fitting problem, i.e., performing excellently on training datasets but very poorly on test datasets, if the training dataset does not {contain} enough samples to represent all possible scenarios. 	
	Due to the {dynamic} nature of the wireless environment, the channel conditions may vary significantly over time.
	For example, {a moving bus} may change channel conditions from LoS to NLoS and vice versa.
	Consequently, real-time data may greatly differ from training data, making these above problems more {severe}.
	Third, wireless channel conditions can also be very different at different areas due to their landscapes, so the DL model trained at one area may not perform well at other areas.
	As such, different sites may need training different models from scratch, which is time-consuming and costly~\cite{hospedales2021meta}.	
	For that, the next section will discuss our approach based on meta-learning that can effectively alleviate these challenges.  
	
	\section{AmB Signal Detector based on Meta-Learning}
	Recently, meta-learning is emerging as a promising approach to address the shortcomings of DL discussed in the previous section. 
	The reason is that meta-learning has an ability of learning to learn, i.e., self-improving the learning algorithm~\cite{hospedales2021meta}.
	{The main idea of meta-learning is to train the model on a collection of tasks, enabling it to acquire generalization capabilities. 
		As a result, the trained model can quickly perform well in a new task only after a few update iterations, even when provided with a small dataset specific to the new task~\cite{finn2017model}.}
	The meta-learning algorithms can be classified into two groups~\cite{nichol2018first}. 
	The first group aims to {encode} knowledge in special and complex DNN architectures, such as Convolutional Siamese Neural Network~\cite{koch2015siamese}, Matching Networks~\cite{vinyals2016matching}, Prototypical networks~\cite{snell2017prototypical}, and Memory-Augmented Neural Networks~\cite{santoro2016meta}.
	As such, this group {requires} more overhead for operations~\cite{finn2017model}, and thus it may not be suitable for our lightweight framework.	
	In the second group, learning algorithms aim to learn the model's initialization.
	Thus, this approach does not require any additional elements or a particular architecture, making it well fit in our framework.
	Therefore, the model initialization is {proposed to use in} our framework. 
	
	A classical approach of the second group (i.e., learning the model's initialization) is that the model is first trained with a large dataset from existing tasks and then fine-tuned with the new task's data.
	However, this pre-training method could not guarantee {that} the learned initialization is good for fine-tuning, and many other techniques need to be leveraged to achieve good performance~\cite{nichol2018first}.
	Recently, the work in~\cite{finn2017model} proposes Model-Agnostic Meta-Learning (MAML) that directly optimizes model's parameters during the learning process with a collection of similar tasks, thus guaranteeing a good initialization of the model. 
	However, MAML needs to perform second-order derivatives through the optimization process, {hence requires high complexity.}
	To overcome this problem, the Reptile algorithm is proposed in~\cite{nichol2018first}.
	This algorithm only uses a conventional DL algorithm, such as stochastic gradient descent (SGD) and Adam, making it {less computationally demanding while still maintaining a similar performance level of MAML~\cite{nichol2018first}.}
	For that reason, this work adopts the Reptile algorithm. 
	\begin{algorithm}[t]
		\caption{The Meta-learning based AmB Signal Detection}
		\label{alg:meta}
		\begin{algorithmic}[1]
			\STATE Initialize model $f$ with parameters $\theta$. 		
			\FOR{\textit{episode = 1, 2, \dots}}
			\STATE Sample task $\tau$ uniformly, and obtain the corresponding training dataset $\mathcal{D}_\tau$
			\FOR{step = 1 to $p$}
			\STATE Based on dataset $\mathcal{D}_\tau$, perform a step of stochastic optimization (SGD or Adam), starting with parameters $\theta$, resulting with parameters \mbox{$\bar{\theta}=U^p_\tau(\theta)$}.				
			\ENDFOR
			\STATE Update $\theta \leftarrow \theta + \eta (\bar{\theta}-\theta)$
			\ENDFOR
		\end{algorithmic}
	\end{algorithm}
	The detail of the meta-learning-based AmB signal detection is provided in Algorithm~\ref{alg:meta}.
	In particular, a model $f$ initiated with parameters $\theta$ is trained with a set $\mathcal{T}$ of similar tasks~\cite{finn2017model, nichol2018first}.
	{As an example, let us consider a set of image classification tasks. 
		The first task involves classifying images into different animal categories, specifically the tiger, dog, and cat classes. 
		On the other hand, the second task focuses on classifying images into another set of animal categories, namely the elephant, bear, and lion classes.}	
	Similarly, this work considers a set of AmB signal detection tasks, each under a particular channel model, e.g., Rayleigh, WINNER II, and Rician.
	
	Generally, meta-learning comprises two nested loops. 
	In the inner loop, a base-learning algorithm (e.g., SGD or Adam) solves task $\tau$, e.g., detecting the AmB signal under the Rician channel.
	The objective of the inner loop's learning is to minimize the loss $\mathcal{L}_\tau$. 	
	Note that a base-learning algorithm can be a typical DL algorithm, such as stochastic gradient descent (SGD).
	Thus, the inner loop is similar to the training process of the proposed DL-based AmB signal detector discussed in Section~\ref{subsec:DL}.
	After $p$ steps of learning, the resulting parameters is denoted by \mbox{$\bar{\theta}=U^p_\tau(\theta)$}, where $U^p_\tau(\cdot)$ denotes the update operator.		
	Then, at the outer loop, a meta-learning algorithm (e.g., Reptile) updates the model parameters as follows:
	\begin{equation}
		\label{eq:outer_update}
		\theta \leftarrow \theta + \eta (\bar{\theta}-\theta),
	\end{equation}
	where $\eta$ is the outer step size controlling how much the model parameters are updated.
	By doing so, the generalization in learning is improved.
	Note that if $p=1$, i.e., performing a single step of gradient descent in the inner loop, Algorithm~\ref{alg:meta} becomes a joint training on the mixture of all tasks, which may not learn a good initialization for meta-learning~\cite{nichol2018first}.
	
	\section{Simulation Results}
	\label{sec:evaluation}
	\subsection{Simulation Settings}
	To evaluate our proposed solution, we conduct simulations in various scenarios {under} the following settings unless otherwise stated.
	For the transmission aspect, the AmB rate is $50$ times lower than the transmitter-to-receiver link rate, i.e., \mbox{$N=50$}, and the AmB message consists of \mbox{$I=100$} bits.	
	Since the transmitter-receiver link SNR $\alpha_{dt}$ significantly depends on many factors, such as transmit power, antenna gain, and path loss, we vary it from $1$~dB to $9$~dB in the simulations to examine our proposed solution in different scenarios.	
	The tag-receiver link SNR $\alpha_{bt}$ is often quite weak, and thus it is set to $-10$~dB~\cite{zhang2018constellation}.	
	
	In this work, the system is simulated using Python, and the DL model is built using the Pytorch library.
	Specifically, the Rayleigh fading follows the zero-mean and unit-variance CSCG distribution, as suggested in~\cite {zhang2018constellation}.
	Note that in our system, the AmB tag does not know the transmitter's signals (i.e., RF resource signals) sending from the transmitter to the receiver. 
	In addition, this work focuses on signal detection for the AmB link.
	As such, the RF resource signals can be assumed to follow the zero-mean and unit-variance CSCG distribution, similar to those in~\cite{zhang2018constellation, guo2018noncoherent, guo2018exploiting}.
	For evaluating different aspects of our proposed solution, we use three metrics, i.e., maximum achievable AmB rate and BER for the transmission performance and guessing entropy for the security analysis.
	
	\subsection{Maximum Achievable Backscatter Rate}
	\begin{figure}
		$\begin{array}{cc}
			\centering
			\includegraphics[width=0.48\linewidth]{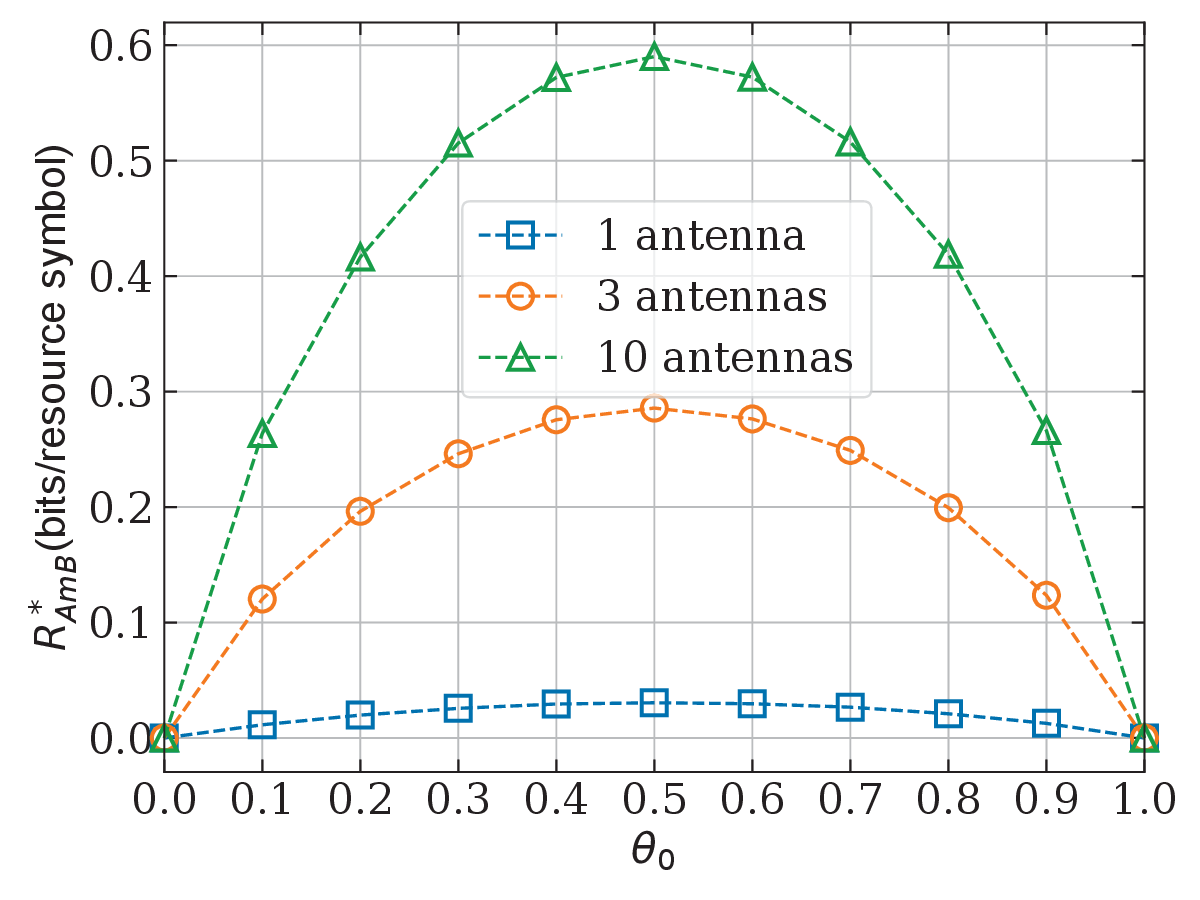} 
			&\includegraphics[width=0.48\linewidth]{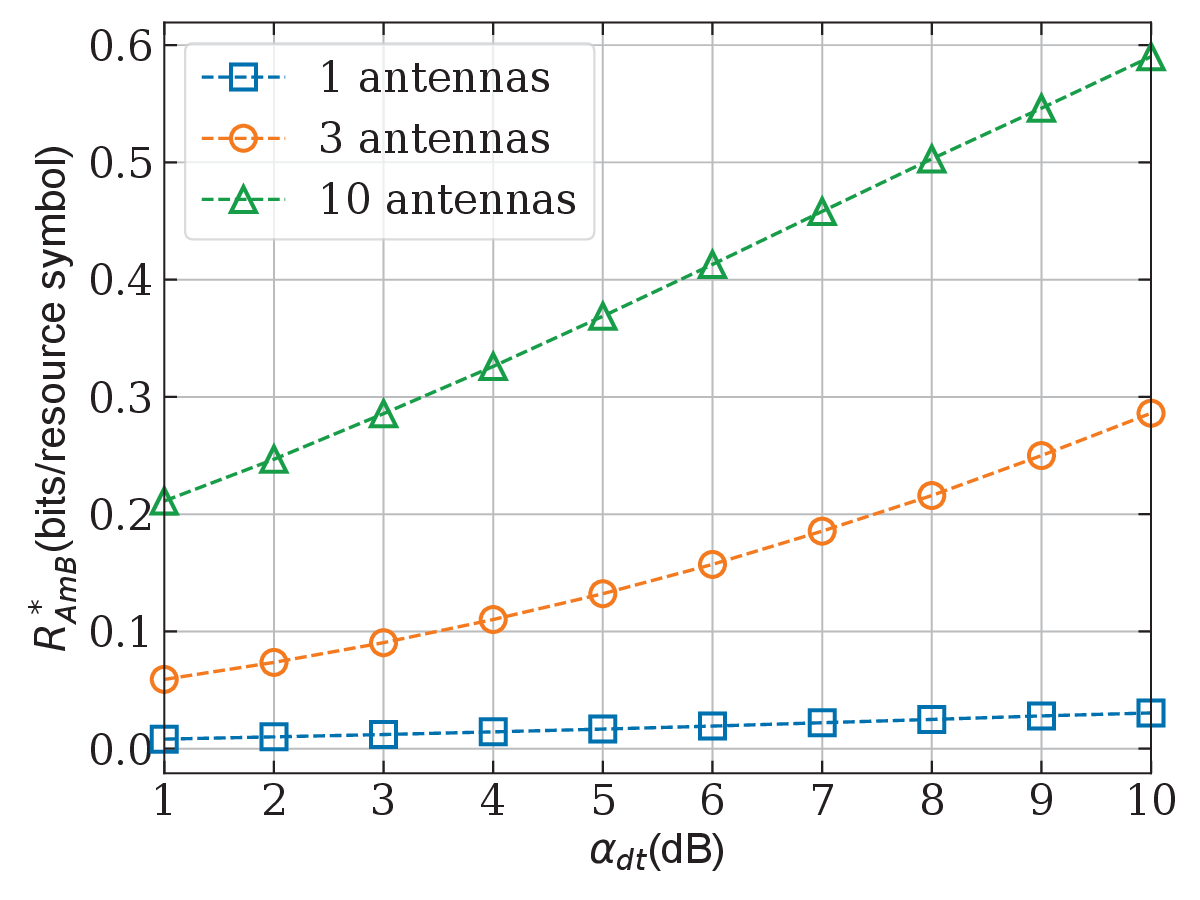} \\
			\text{(a)} & \text{(b)}
			\vspace{-10pt}
		\end{array}$
		\caption{The maximum achievable AmB rate when varying (a) $\theta_0$ and (b) the transmitter-to-receiver link SNR, i.e., $\alpha_{dt}$.}
		\label{fig:MaximumRate_vary_theta0_SNR}
		\vspace{-10pt}
	\end{figure}
%
	To get insights on the maximum achievable AmB rate $R_{AmB}^*$ of the proposed system, we perform simulations in three settings, where the receiver has $1, 3$ and $10$ antennas, as shown in Fig.~\ref{fig:MaximumRate_vary_theta0_SNR}.
	First, in Fig.~\ref{fig:MaximumRate_vary_theta0_SNR}(a), the prior probability of backscattering bit ``0'', i.e., $\theta_0$, is varied to observe the maximum achievable AmB rate, i.e., the idealized throughput $R_\text{AmB}^*$ which is given by Theorem~\ref{theo:maxrate}.
	The results are obtained by $10^6$ Monte Carlo runs. 
	Fig.~\ref{fig:MaximumRate_vary_theta0_SNR}(a) shows that the maximum achievable AmB rate $R_{AmB}^*$ {increases} when the receiver has more antennas.
	This is because the receiver can achieve a higher gain with more antennas. 
	As such, the received backscatter signals are enhanced, thereby reducing the impacts from fading and interference originating from the direct link and the surrounding environment.
	Importantly, $R_{AmB}^*$ is maximized at \mbox{$\theta_0=0.5$}.
	Therefore, we set \mbox{$\theta_0=0.5$} in all {the rest of} simulations.
	Next, we vary the transmitter-to-receiver link SNR, i.e., $\alpha_{dt}$, from $1$~dB to $9$~dB, as shown in Fig.~\ref{fig:MaximumRate_vary_theta0_SNR}(b).
	It can be observed that as the active link SNR $\alpha_{dt}$ increases, the maximum achievable backscatter rate $R_{AmB}^*$ increases.
	The reason is that when the transmitter increases its transmit power, the signal arriving at the AmB tag is stronger, making the backscattered signals stronger.

	\subsection{Anti-eavesdropper and Security Analysis}
	{It is worth noting that our proposed framework can counter eavesdropping attacks in two folds.
		Firstly, a part of the original message is hidden in the AmB channel, appearing as background noise for conventional signal detectors.
		Secondly, the proposed message encoding mechanism makes reconstructing the original message non-trivial unless the eavesdropper knows this and at the same time it can decode information from both active and AmB channels.
		
		In particular, without knowledge about the system in advance (i.e., the settings of AmB transmission), the eavesdropper is not even aware of the existence of the AmB message, thus guaranteeing the confidentiality of the original message.
		It is worth emphasizing that even if the eavesdropper is aware of the AmB message, it is still challenging to capture the AmB message.
		Suppose that the eavesdropper leverages the AmB signal detector circuit proposed in~\cite{liu2013ambient}, which requires different values of resistors and capacitors in the circuit for different backscatter rates.
		As such, if the eavesdropper deploys the AmB circuit but does not know the exact backscatter rate, it still cannot decode the backscatter signals~\cite{liu2013ambient}.
		Iteratively testing every possible rate is impractical.
		It can be considered that the eavesdropper can use MLK- or DL-based detectors for detecting AmB signals.
		However, MLK-based approaches require complete information on signal distribution and perfect CSI, while DL-based methods need a large amount of data and time to detect AmB signals properly.
		Given the above, the probability that the eavesdropper can successfully capture the AmB message is marginal in practice.
		
		To further evaluate the security of our proposed system, we consider the worst-case in which the eavesdropper knows the exact backscatter rate in advance, and thus they can capture the AmB message as well as the active message.		
		However, since the message encoding technique is unknown, the eavesdropper still does not know how to construct the original message based on the active and AmB messages.
		They only know that combining these two messages can decode the original message.
		As such, they must determine the position of all $I$ bits of the AmB message in the original message.
		To quantify the security of our proposed anti-eavesdropping solution in this case, this paper considers the guessing entropy {metric}~\cite{boztas1997comments}.}
	
	\begin{figure}[!]
		\centering
		\includegraphics[width=0.5\linewidth]{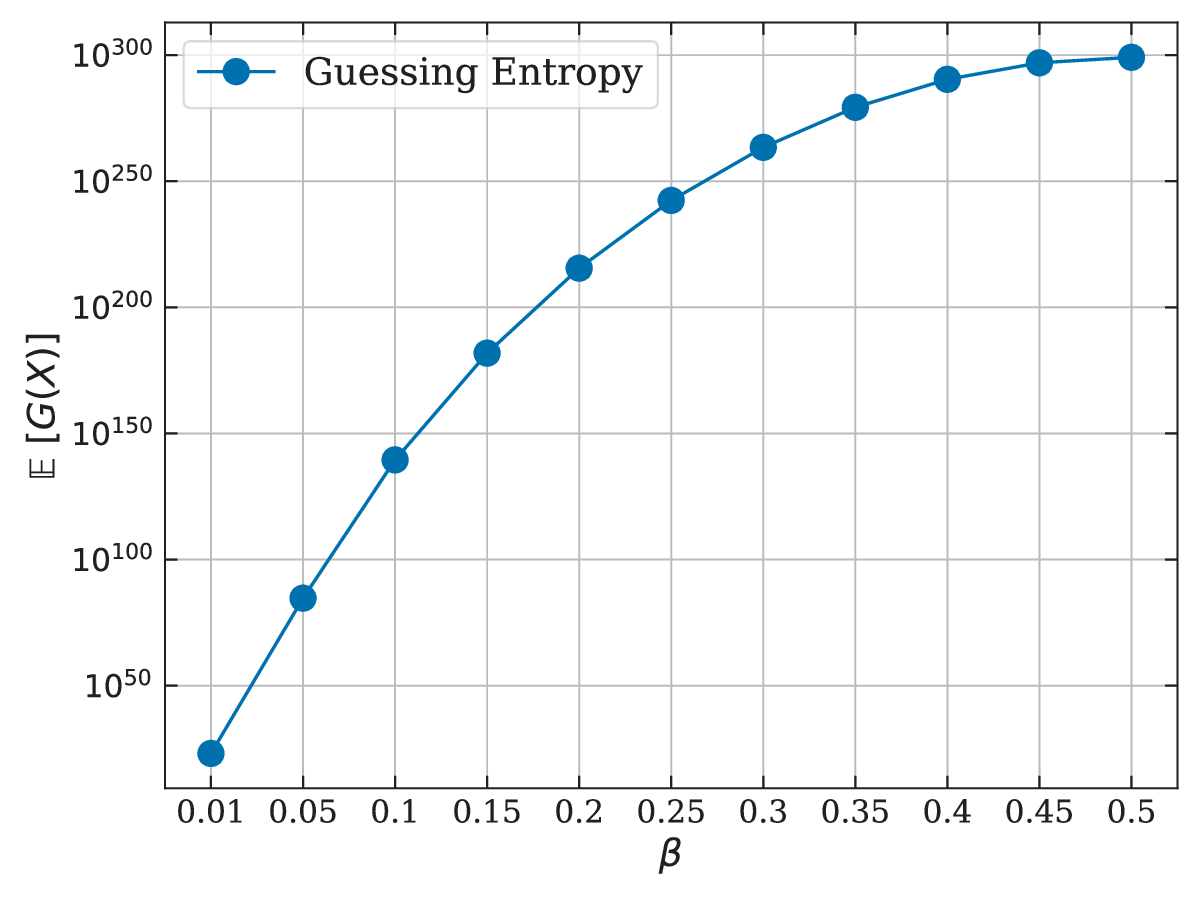}
		\caption{The upper bound of the expected number of guesses, i.e., $\mathbb{E}[G(X)]$,  vs. the message splitting ratio $\beta$.}
		\label{fig:guessing_entropy}
		\vspace{-20pt}
	\end{figure}
	Suppose that the original message has $P$ bits, the probability that the eavesdropper successfully finds the positions of $I$ bits in the original message by guessing is given by 
	\begin{equation}
		\label{eq:key_prob}
		p_I = \frac{1}{C^P_I} = \frac{I!(P-I)!}{P!}.
	\end{equation} 
	To correctly identify the position of $I$ bits, the eavesdropper needs to ask a number of questions of the kind ``does this correct?''. 
	Suppose that the eavesdropper follows an optimal guessing sequence, the average number of these questions defines the guessing entropy~\cite{boztas1997comments}.
	We can consider that the position of $I$ bits is analogous to the key $k_I$ for reconstructing the original message.	
	Let $X$ denote the random variable corresponding to the selection of key $k_I$ in the key set $\mathcal{K}$, and $X$ is determined by probability distribution $\mathcal{P}_k$.
	Note that in this work, the key set consists of all possible keys, and thus the size of $\mathcal{K}$ is $C^P_I$.
	One of the optimal brute-force attacks is the scheme where the eavesdropper knows $\mathcal{P}_k$, and {it sorts} the key set $\mathcal{K}$ in the descending order of key selection probability $p_k$ to form a guessing sequence $\bar{\mathcal{K}}$~\cite{boztas1997comments}.
	Then, it iteratively tries a key in the sorted key set $\bar{\mathcal{K}}$. 
	For that, the guessing entropy is defined as follows:
	\begin{equation}
		G(X) = \sum_{1\leq i \leq |\bar{\mathcal{K}}|} i p(x_i),
	\end{equation}
	where $i$ in the index in $\bar{\mathcal{K}}$, and $p(x_i)$ corresponds to the selection probability of a key at index $i$.
	In~\cite{boztas1997comments}, the upper bound of the expected number of guesses, i.e., $\mathbb{E}[G]$, for the eavesdropper with the optimal guessing sequence is given as follows:
	\begin{equation}
		\mathbb{E}[G(X)] \leq \frac{1}{2}\left[\sum_{i=1}^{|\bar{\mathcal{K}}|}\sqrt{p(x_i)}\right]^2 + \frac{1}{2}.
	\end{equation}

	In Fig.~\ref{fig:guessing_entropy}, we vary the message splitting ratio \mbox{$\beta = P/I$} to observe the guessing entropy of our framework.
	Here, \mbox{$\beta=0.1$} means that	$10\%$ of original message bits are transmitted by the AmB tag, and the remaining bits are transmitted by the transmitter.
	Since the AmB rate is significantly lower than the transmitter rate~\cite{guo2018noncoherent, zhang2018constellation}, the message splitting ratio $\beta$ is less than $0.5$ in practice.
	As shown in Fig.~\ref{fig:guessing_entropy}, when the ratio $\beta$ increases from $0.01$ to $0.5$, the guessing entropy also increases.
	It is stemmed from the fact that as the size of the AmB message increases, the probability of guessing successfully in one trial decreases, as implied by \eqref{eq:key_prob} when \mbox{$\beta < 0.5$}.
	It is worth mentioning that guessing entropy is the average number of questions asked by the eavesdropper to successfully construct the original message. 
	As such, the greater the value of $\mathbb{E}[G(X)]$, the higher the security level is. 
	
	\subsection{The Learning Process of DL-based AmB Signal Detector}
	\label{subsec:learning_process}
		
	The settings for the DL-based approaches are as follows.
	Note that since the architecture of DNN can significantly affect the performance of the DL-based AmB signal detector, it must be designed thoughtfully.
	For example, a DNN with more layers may perform better but demands more resources and takes longer time for training.
	Thus, this work designs a simple and lightweight DNN while still achieving good performance of detecting AmB signal.
	Particularly, our DNN has an input layer, four {FC} layers, and an out layer.
	The number of neurons in the input layer depends on the dimension of the input data $\mathbf{d}$, which is \mbox{$M\!\times\!M\!\times\!2\!\times3$} where $M$ is the number of the receiver's antennas.
	The first, second, and third FC layers have $600, 1000,$ and $600$ neurons, respectively.
	The training dataset is obtained from $10^4$ Monte Carlo runs.
	\begin{figure}[!]
		\centering
		\includegraphics[width=0.5\linewidth]{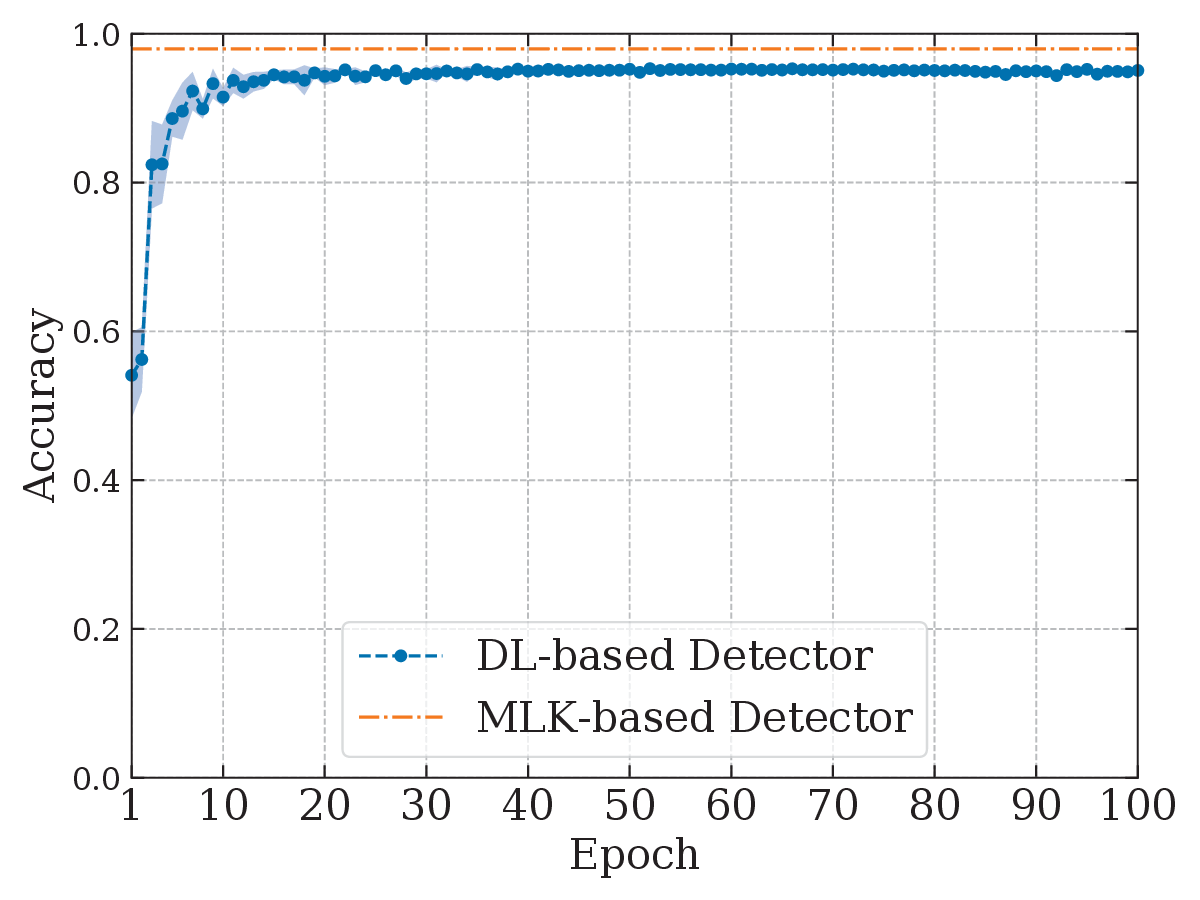}
		\caption{The convergence of DL-based AmB signal detector when \mbox{$\alpha_{dt}=1$}~dB and the receiver has $10$ antennas.}
		\label{fig:learning process}
		\vspace{-10pt}
	\end{figure}
	
	Fig.~\ref{fig:learning process} shows the convergence rate of our proposed DL-based AmB signal detector when the direct link SNR \mbox{$\alpha_{dt}=1$}~dB and the receiver has $10$ antennas.
	To train the DNN, we use the SGD with the learning rate of $0.001$~\cite{goodfellowdeep2016}.
	The batch size is set to $1,000$ data points, and thus each training epoch consists of $10$ learning iterations.
	In Fig.~\ref{fig:learning process}, the accuracy of the MLK-based detector is also presented as a baseline.
	It can be observed that after $30$ learning {epochs}, our proposed DL-based detector converges into a reliable model that can achieve an accuracy close to that of the MLK-base detector.
	Note that MLK-based detectors are considered the optimal signal detection scheme, but it is impractical due to high computing complexity and the requirement of perfect CSI~\cite{wubben2004near}.
	
	\subsection{BER Performance}
	\label{subsec:BER_performance}
	
	Now, we investigate the system's BER performance when varying the direct link (i.e., transmitter-to-receiver) SNR $\alpha_{dt}$ from $1$~dB to $9$~dB.
	The training process of our proposed DL-based detector is conducted in the same way described in~\ref{subsec:learning_process} with two settings, i.e., providing (i) the estimated CSI based on pilot bits, namely DL-eCSI, and (ii) the perfect CSI, namely DL-pCSI.
	We perform simulations with three antenna configurations for the receiver, i.e., $1$, $3$, and $10$ antennas.	 
	To obtain reliable results, both the MLK-based and the DL-based detectors are evaluated by performing Monte Carlo fashion with $10^6$ runs.
	Note that this work only focuses on the AmB link (i.e., the tag-receiver link), and thus we only obtain the BER of the AmB link since the BER of the {direct (transmitter-receiver)} link can be close to zero {with advanced modulation and channel coding techniques}.
	
	Fig.~\ref{fig:BER}(a) shows that when the direct link SNR, i.e., $\alpha_{dt}$, increases from $1$~dB to $9$~dB, the BERs of all approaches decrease. 
	The rationale behind is that the stronger the transmitter's active signals are, the stronger the signals backscattered by the AmB tag (i.e., AmB signals) are, so the lower the system BER is.
	It is also observed that the system performance improves (i.e., a decrease in BER) as the number of receiving antennas increases. 
	This is because the received signals at the receiver are enhanced by antenna gain, which is typically proportional to the number of antennas.
	Similar to the observation in Section~\ref{subsec:learning_process}, the BERs of our DL-based detector are close to those of the MLK-detector, an optimal signal detector, in all cases with 1, 3, and 10 antennas.
	Thus, these clearly show the effectiveness of our proposed DL-based approach.
	\begin{figure*}[!t]
		\centering
		\includegraphics[width=0.99\linewidth]{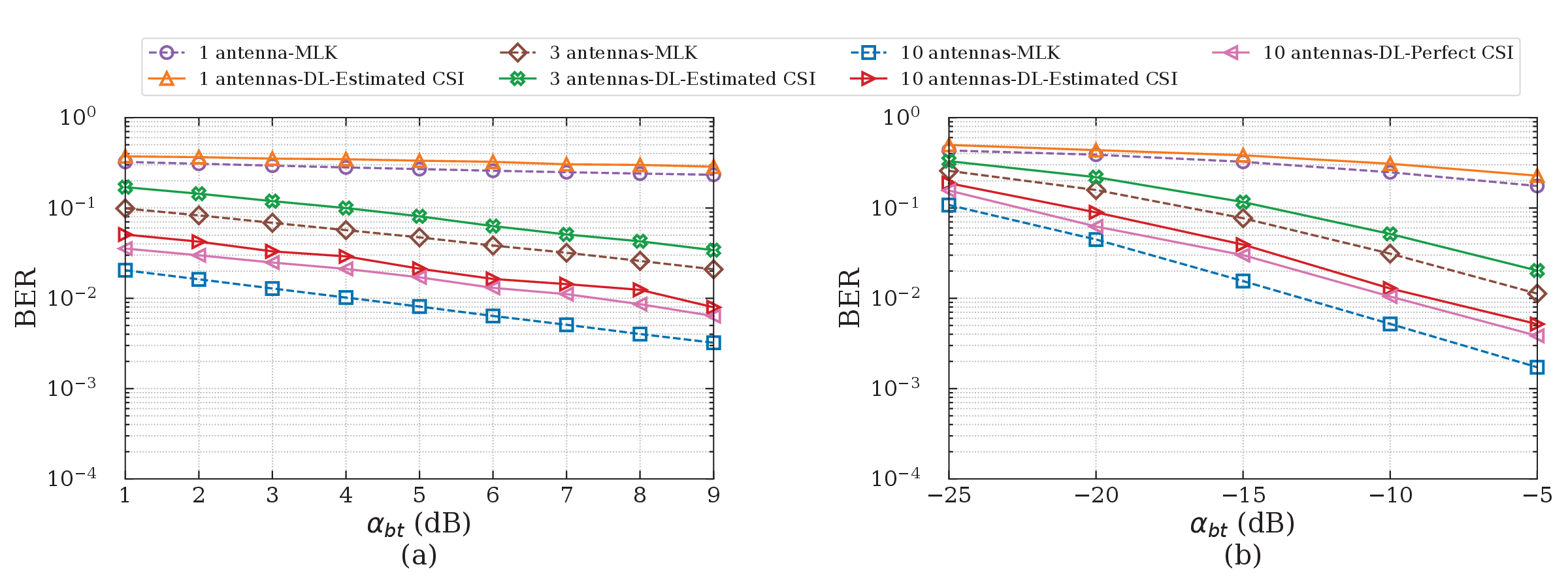}
		\caption{Varying (a) the transmitter-receiver SNR $\alpha_{dt}$ and (b) the tag-receiver SNR $\alpha_{bt}$.}
		\label{fig:BER}
		\vspace{-10pt}
	\end{figure*}
	
	%
	
	Next, we vary the tag-receiver SNR, i.e., $\alpha_{bt}$, from $-25$~dB to $-5$~dB, while the transmitter-receiver SNR $\alpha_{dt}$ is $7$~dB, as shown in Fig.~\ref{fig:BER}(b).
	Similar to observations in Fig.~\ref{fig:BER}(a), when $\alpha_{bt}$ and the number of antennas {increase}, the BERs of all approaches reduce. 	
	The rationale behind this is that when the received signal is stronger and the antenna gain is higher, the detector can attain a lower BER, indicating better performance.
	Again, the gap observed in BERs between our proposed DL-based detector and the MLK method are marginal.
	Interestingly, the DL-eCSI can achieve comparable BER results compared with that of DL-pCSI, as shown in Figs.~\ref{fig:BER}(a) and (b). 
	Thus, these results show that our proposed DL-based detector can perform well in practice {with only estimated CSI}.  
	
	\subsection{Meta-Learning for AmB Signal Detection}
	Now, we evaluate the proposed meta-learning approach for the AmB-signal detector with the following setups.
	We consider that the target task is detecting AmB signals under Rayleigh fading to better demonstrate the comparison between meta-learning-, DL-, and MLK- based approaches.
	The set of tasks for the learning process, i.e., training tasks, consists of (i) detecting AmB signals under Rician fading and (ii) detecting AmB signals under {WINNER} II fading~\cite{kyosti2007winner}.
	Algorithm~\ref{alg:meta} is used to train the DNN model, i.e., meta-model, with the training tasks' datasets.	
	Here, each training task's dataset contains $500$ data points, and the outer step size $\eta$ is set to $0.1$, similar to~\cite{nichol2018first}.	
	Then, the trained meta-model is trained with the dataset collected from the target task (i.e., detecting AmB signals under Rayleigh fading) by using SGD.
	For these simulations, we select two types of baselines{: (i) DL- and (ii) }MLK-based approaches.
	Three DL models are trained with different numbers of data points, including $50$, $10^3$, and $10^4$, in a similar procedure in Section~\ref{subsec:learning_process}. 
	In all approaches, there are $10$ antennas in the receiver.
	
	\begin{figure}[t]
		\centering
		\includegraphics[width=0.5\linewidth]{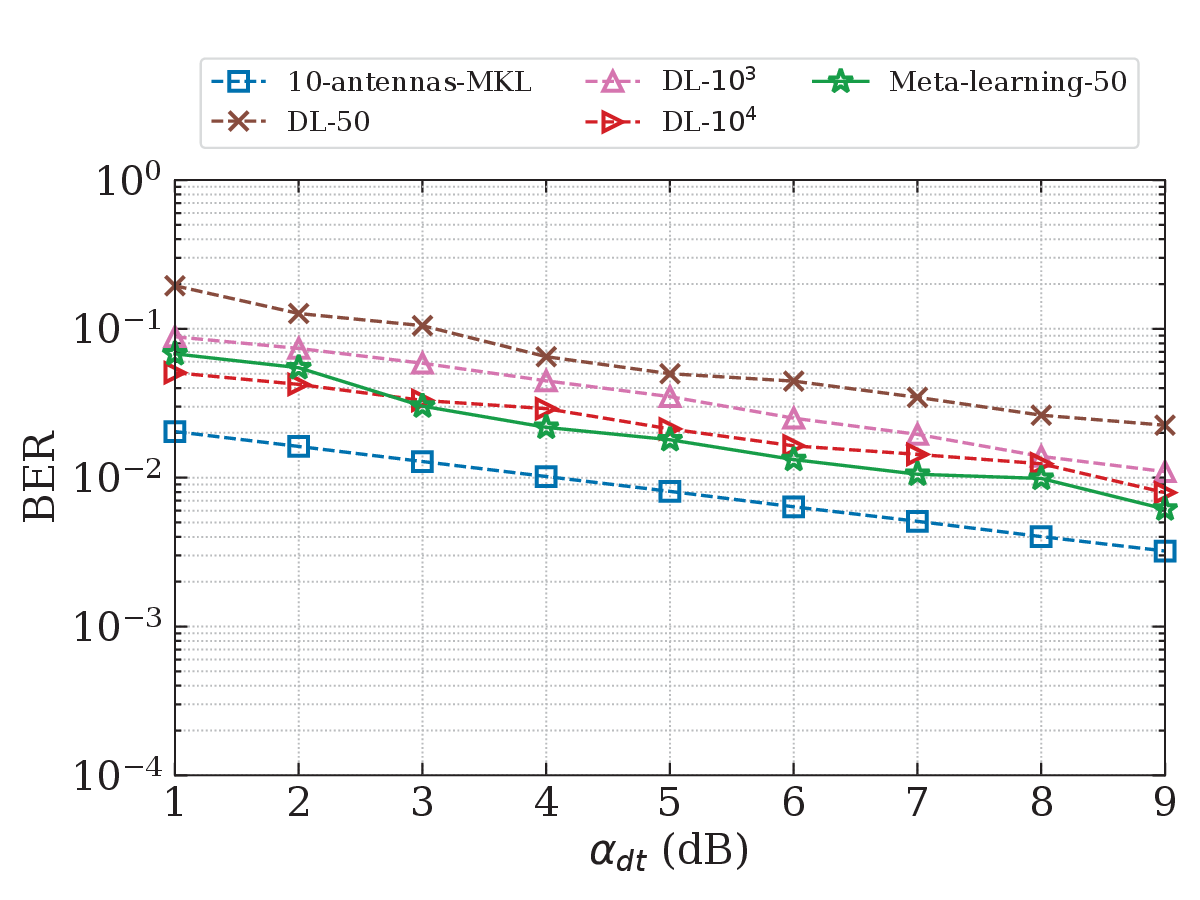}
		\caption{The BER performance of different learning approaches.}
		\label{fig:meta_BER}
		\vspace{-20pt}
	\end{figure}	
	
	First, we compare the BER performance between our proposed meta-learning-based approach and other baselines when varying the transmitter-receiver SNR $\alpha_{dt}$, as shown in Fig.~\ref{fig:meta_BER}.
	Similar to the observations in {Figs}.~\ref{fig:BER}(a) and (b)
, as the signal's strength increases (i.e., an increase of SNR), all approaches achieve a lower BER, indicating better performance.
It is observed that a DL model trained with more data points results in a better performance.
In particular, among DL-based approaches, DL with $10^4$ data points achieves the lowest BER, and DL with $50$ data points gets the highest BER.
The reason is that more data points can gain more knowledge to the DL model. Thus, it can reduce the over-fitting issue, i.e., performing well on a training dataset but poor on a test dataset, thereby improving the system performance~\cite{goodfellowdeep2016}.	
Interestingly, when the transmitter-receiver SNR \mbox{$\alpha_{dt}\geq3$}~dB, meta-learning with only $50$ data points can achieve the lowest BER compared with DL-based approaches that require up to $10^4$ data points.
Whereas, when $\alpha_{dt}$ is low (less than $3$~dB, meta-learning's BER is slightly higher than that of the DL with $10^4$ data points but still lower than those of DL with $10^3$ and $50$ data points.
These observations are stemmed from the fact that meta-learning has a generalization ability that can help a DNN model to learn a new task quickly with few data points~\cite{nichol2018first}. 

To further investigate the reliability of these learning approaches, we run each learning approach $20$ times to obtain the standard deviation of its BER results, as shown in Fig.~\ref{fig:meta_std}.
The simulation settings are set as those in Fig.~\ref{fig:meta_BER}.
Generally, the standard deviations of all approaches decrease as the transmitter-receiver SNR $\alpha_{dt}$ increases from $1$~dB to $10$~dB.
This is because the received signals contain more noise at a lower SNR, making the learning more challenging and leading to high deviation of results, i.e., more unstable results. 
In particular, at \mbox{$\alpha_{dt}=1$}, the standard deviation of DL with $50$ training data points is about $10$ times higher compared with that of the best learning-based approach i.e., DL with $10^4$ data points.
Again, when \mbox{$\alpha_{dt}\geq3$}~dB, the proposed meta-learning-based approach achieves {comparable} results as those of DL with $10^4$ data points.
Notably, when \mbox{$\alpha_{dt}\geq8$}~dB, the results of meta-learning and DL with $10^4$ data points close to the results of MLK.
When the \mbox{$\alpha_{dt}$} decreases to less than $3$~dB, the meta-learning still obtains very good results, one of the best approaches in terms of reliable performance.

\begin{figure}[t]
	\centering
	\includegraphics[width=0.5\linewidth]{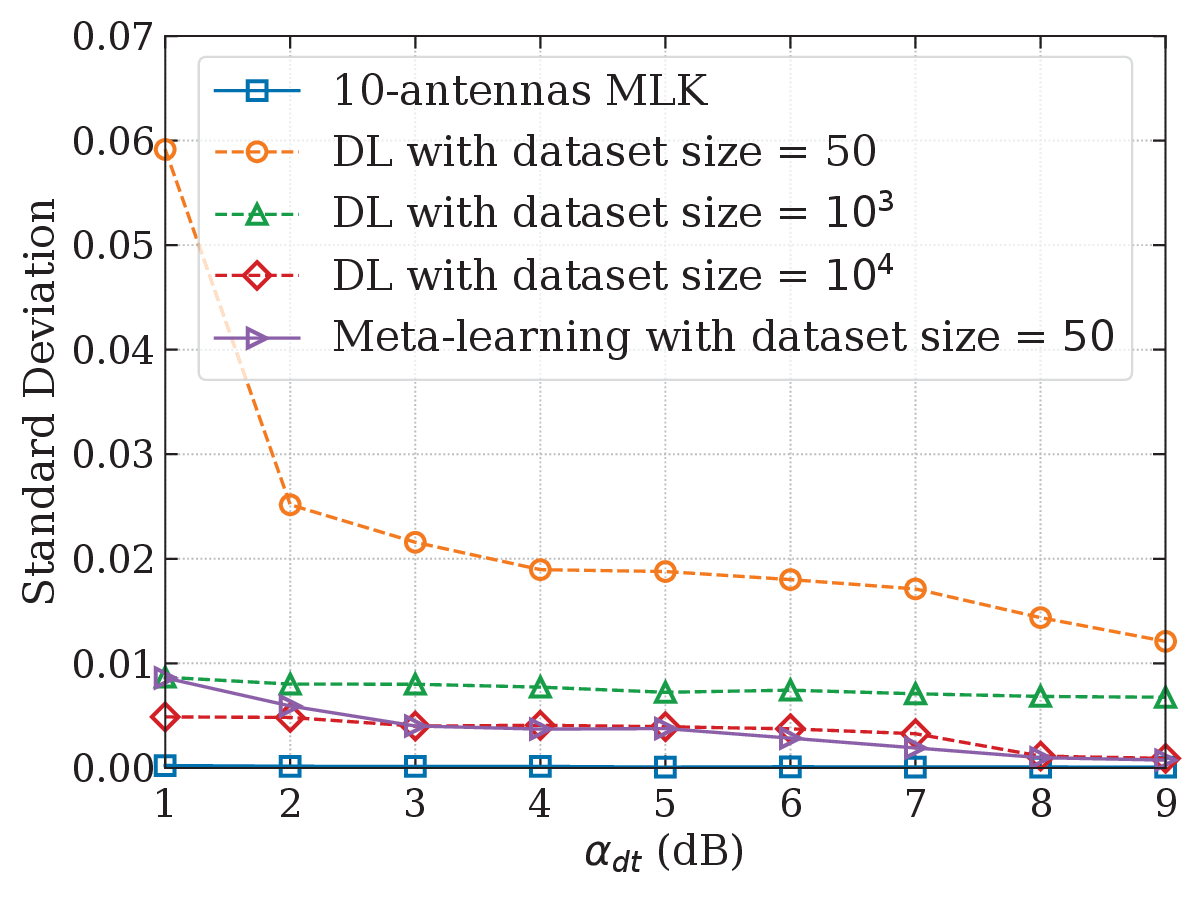}
	\caption{Reliability of learning process with different training datasets' sizes.}
	\label{fig:meta_std}
	\vspace{-10pt}
\end{figure}
\begin{figure}[t]
	\centering
	\includegraphics[width=0.5\linewidth]{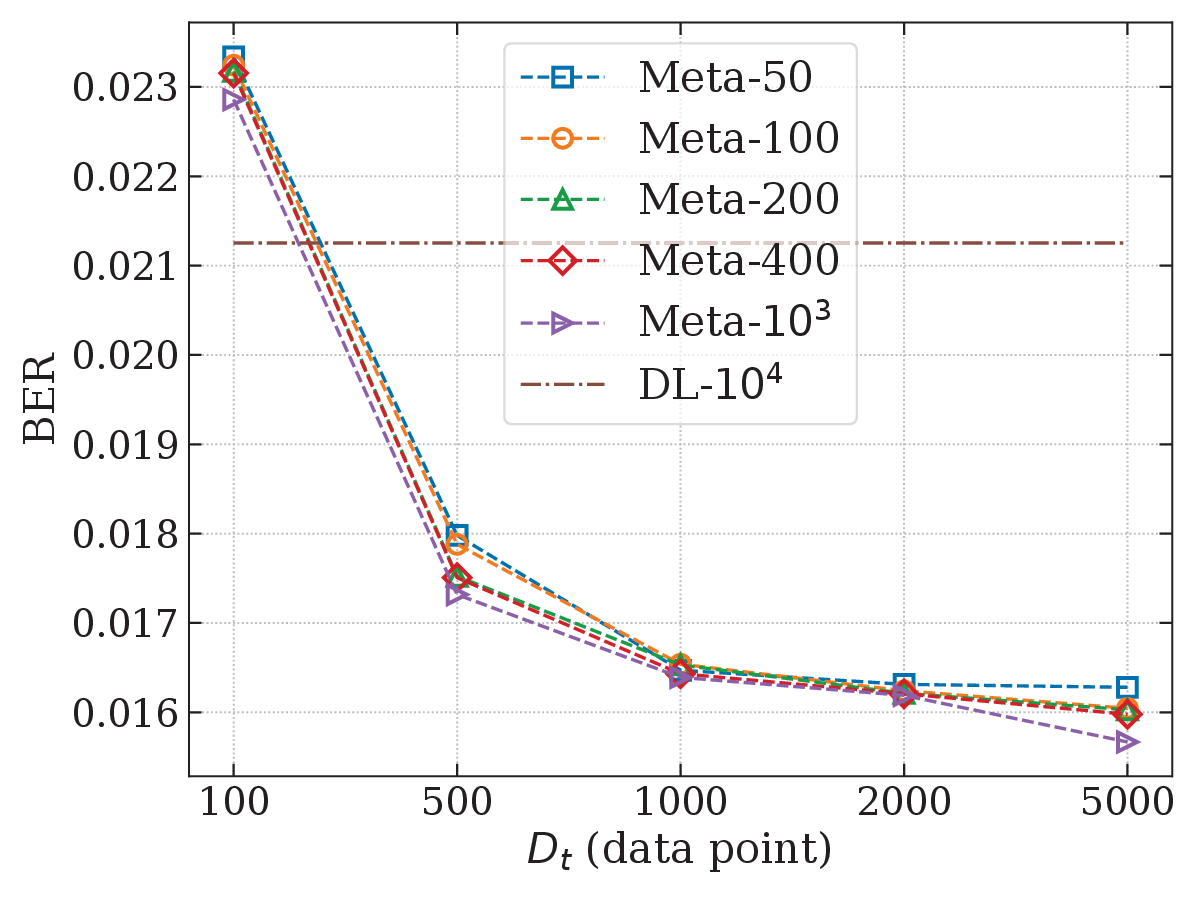}
	\caption{Varying the size of training tasks' datasets.}
	\label{fig:meta_BER_SNR5}
	\vspace{-10pt}
\end{figure}

Finally, to get insight into how {the amount of training data} can impact meta-learning, we set the size of each dataset collected from a training task to $D_t$ and then vary $D_t$ from $100$ to $500$ data points.
	Here, the transmitter-receiver SNR $\alpha_{dt}$ is set to $5$~dB.
	Fig.~\ref{fig:meta_BER_SNR5} {shows} the BER performance of models trained by meta-learning with $50$, $100$, $200$, $400$, and $10^3$ data points from the target task, namely Meta-50, Meta-100, Meta-400, and Meta-$10^3$.
	For comparing with the DL-based approach, this figure also presents the result of the model trained from scratch by SGD (as in Section \ref{subsec:learning_process}) at \mbox{$\alpha_{dt}=5$} with $10^4$ data points of the target task, namely \mbox{DL-$10^4$}.
	Similar to observations in Fig.~\ref{fig:meta_BER}, all approaches achieve lower BERs as the training data size increases, as shown in Fig.~\ref{fig:meta_BER_SNR5}.
	Interestingly, when the meta-learning training data points are adequate (e.g., $D_t\leq500$), even with only $50$ target task's data points, meta-learning approaches outperform the DL approach training with $10^4$ target task's data points.
	Given the above, it clearly shows the superiority of meta-learning to DL approaches in terms of data efficiency, making it more {applicable} in practical systems.

\section{Conclusion}
\label{sec:conclusion}
This work has introduced a novel anti-eavesdropping framework that can enable secure wireless communications leveraging the low-cost and low-complexity AmB technology. 
In particular, an original message {was} divided into two messages, and they are transmitted over two channels, i.e., direct transmission channel and AmB communication channel.
Since the AmB tag only backscatters the RF signals to transmits data, instead of generating active signals, {the eavesdropper is unlikely aware of} the existence of AmB transmissions.
Even if eavesdroppers can capture AmB signals, our proposed splitting message introduces much more difficulties for the eavesdroppers to derive the original message.
To effectively decode AmB signals at the receiver, we have developed a signal detector based on DL.
Unlike MLK-based solutions, our detector {did} not rely on a complex mathematical model {nor} require perfect CSI.	
To make the DL-based detector to quickly achieve good performance in new environments, we have developed the meta-learning technique for the training process.   
The simulation results have shown that our proposed approach can not only ensure the security of the data communications in the presence of eavesdroppers but also can achieve the BER performance that is comparable to the MLK-based detector, which relies on perfect CSI.

\appendices
\section{The proof of Theorem~\ref{theo:maxrate}}
\label{appendix:maximumbackscatterrate}
In this appendix, we {prove} Theorem~\ref{theo:maxrate} in a similar way as in~\cite{guo2018noncoherent}. 
	Since the mutual information between the AmB tag's state $e$, i.e., \mbox{$\mathcal{M}(e;\mathbf{y})$} and the received signals $\mathbf{y}$ defines the achievable AmB rate, we can obtain the maximum achievable AmB rate by~\cite{guo2018noncoherent}
	\begin{equation}
		\label{eq:Rb}
		R_{AmB}^* = \mathbb{E}[\mathcal{M}(e, \mathbf{y})],
	\end{equation}
	where 
	\begin{equation}
		\label{eq:I}
		\mathcal{M}(e, \mathbf{y}) = Z(\theta_0) - \mathbb{E}_{\mathcal{V}}[C(e|\mathcal{V})].
	\end{equation}
	In~\eqref{eq:I}, $Z(\theta_0)$ is the binary entropy function defined by~(\ref{eq:Ctheta0}), and $C(e|\mathcal{V})$ is the conditional entropy of $e$ given $\mathcal{V}$.
	\begin{equation}
		\label{eq:Ctheta0}
		Z(\theta_0) \triangleq -\theta_1 \log_2 \theta_1-\theta_0 \log_2 \theta_0 .
	\end{equation}
	Note that since $Z(\theta_0)$ does not depend on the channel coefficients, the maximum achievable AmB rate $R_{AmB}^*$ is given by
	\begin{equation}
		R_{AmB}^* = \mathbb{E}[\mathcal{M}(e, \mathbf{y})] = Z(\theta_0) - \mathbb{E}_{\mathcal{V}}[C(e|\mathcal{V})].
	\end{equation}
	In addition, given $\tilde{\mathbf{y}}$, the posterior probability of $e$ is calculated by
	\begin{equation}
		p(e=j|\mathcal{V}) = \frac{\theta_jp(\mathbf{y}|e=j)}{ \theta_1 p(\mathcal{V}|e=1) +\theta_0p(\mathcal{V}|e=0) },
	\end{equation}
	with $j \in \{0,1\}$. 
	By letting $\mu_j = p(e=j|\mathcal{V})$, the conditional entropy $C(e|\mathcal{V})$ is expressed by
	\begin{equation}
		C(e|\mathcal{V}) = -\sum_{j=0}^{1}\mu_j\log_2\mu_j = Z(\mu_0).
	\end{equation}
	Consequently, we have~\cite{guo2018noncoherent}
	\begin{equation}
		\begin{aligned}
			R_{AmB}^* &= Z(\theta_0) - \mathbb{E}_{\tilde{\mathbf{y}}}[Z(\mu)]\\
			& = Z(\theta_0) - \int_{\mathcal{V}}\big( \theta_1 p(\mathcal{V}|e=1) + \theta_0p(\mathcal{V}|e=0) \big)Z(\mu_0)d\mathcal{V}.
		\end{aligned}
	\end{equation}

\bibliographystyle{IEEEtran}
\bibliography{refs_AmB}

\end{document}